\documentclass[conference]{IEEEtran}
\IEEEoverridecommandlockouts

\usepackage{cite}
\usepackage{amsmath,amssymb,amsfonts}
\usepackage{algorithmic}
\usepackage{graphicx}
\usepackage{textcomp}
\usepackage{xcolor}
\def\BibTeX{{\rm B\kern-.05em{\sc i\kern-.025em b}\kern-.08em
    T\kern-.1667em\lower.7ex\hbox{E}\kern-.125emX}}

\usepackage{times}
\usepackage{booktabs}
\usepackage{bbm}
\usepackage{graphicx}
\usepackage{epstopdf}
\usepackage[linesnumbered,ruled,vlined]{algorithm2e}
\usepackage{booktabs}
\usepackage{url}
\usepackage{multirow}
\usepackage{subfigure}
\usepackage{enumitem}
\usepackage{tabularx}

\newtheorem{definition}{Definition}
\newtheorem{example}{Example}
\newtheorem{observation}{Observation}[section]
\newenvironment{proof}{ {\it Proof:} }{\hfill $\square$\par}  
\newtheorem{corollary}{Corollary }[section]

\newtheorem{lemma}{Lemma}[section]
\newtheorem{property}{Property}[section]

\newcommand{\kw}[1]{{\ensuremath {\mathsf{#1}}}\xspace}

\newcommand{\stitle}[1]{\vspace{1ex} \noindent{\bf #1}}
\long\def\comment#1{}

\def\done{\hspace*{\fill}$\square$\par}

\newcommand{\mdcoreall}{$(l,\delta)$-maximal dense core\xspace}

\newcommand{\mdcorealls}{$(l,\delta)$-maximal dense cores\xspace}
\newcommand{\mdcore}{$(l,\delta)$-\kw{MDC}}
\newcommand{\mdcores}{$(l,\delta)$-\kw{MDCs}}
\newcommand{\mdc}{\kw{MDC}}

\newcommand{\densenode}{$(l,\delta)$-dense node\xspace}
\newcommand{\densenodes}{$(l,\delta)$-dense nodes\xspace}

\newcommand{\msd}{$\mathcal{MSD}$\xspace}

\newcommand{\lmsdnospace}{$\kw{MSD}$}

\newcommand{\lmsdall}{maximum $l$-segment density\xspace}

\newcommand{\mts}{\kw{MTS}}
\newcommand{\mtsnospace}{\kw{MTS}}

\newcommand{\mtsinospace}{$\kw{MTS}_{2l}$}
\newcommand{\mtsistar}{\kw{MTS}_{2l}^*}

\newcommand{\lsdnospace}{\kw{SD}}

\newcommand{\ch}{$\mathcal{CH}$\xspace}
\newcommand{\chnospace}{\mathcal{CH}}

\newcommand{\chess}{\kw{Chess}}
\newcommand{\lkml}{\kw{Lkml}}
\newcommand{\enron}{\kw{Enron}}
\newcommand{\dblp}{\kw{DBLP}}
\newcommand{\youtube}{\kw{YTB}}
\newcommand{\flickr}{\kw{FLK}}

\newcommand{\mathoverflow}{\kw{MO}}
\newcommand{\askubuntu}{\kw{AU}}
\newcommand{\wikitalk}{\kw{WT}}

\newcommand{\as}{{\kw{AS}}\xspace}

\newcommand{\ad}{{\kw{AD}}\xspace}

\newcommand{\skyline}{{\kw{POMDC}}\xspace}
\newcommand{\skylines}{{\kw{POMDCs}}\xspace}

\newcommand{\skylineb}{{\kw{POMDC\text{-}B}}\xspace}

\newcommand{\kcore}{{\kw{KCORE}}\xspace}
\newcommand{\maxdensesub}{{\kw{DENSEST}}\xspace}
\newcommand{\mdcb}{{\kw{MDC\text{-}B}}\xspace}
\newcommand{\mdcplus}{\kw{MDC}\text{+}\xspace}

\begin{document}

\title{Mining Bursting Communities in Temporal Graphs}

\author{{Hongchao Qin$\scriptsize^{\dag}$, Rong-Hua Li$\scriptsize^{\ddag}$, Guoren Wang$\scriptsize^{\ddag}$, Lu Qin$\scriptsize^{\#}$, Ye Yuan$\scriptsize^{\ddag}$, Zhiwei Zhang$\scriptsize^{\ddag}$}
	\vspace{1.6mm}\\
	\fontsize{9}{9}\selectfont\itshape
	$\scriptsize^{\dag}$Northeastern University, China; $\scriptsize^{\ddag}$Beijing Institute of Technology, China; $\scriptsize^{\#}$University of Technology Sydney, Australia \\
	\fontsize{7}{8}\selectfont\ttfamily\upshape
	qhc.neu@gmail.com; \{rhli, wanggr\}@bit.edu.cn; Lu.Qin@uts.edu.au; yuanye@mail.neu.edu.cn; cszwzhang@comp.hkbu.edu.hk
}
\maketitle

\begin{abstract}
Temporal graphs are ubiquitous. Mining communities that are bursting in a period of time is essential to seek emergency events in temporal graphs. Unfortunately, most previous studies for community mining in temporal networks ignore the bursting patterns of communities. In this paper, we are the first to study a problem of seeking bursting communities in a temporal graph. We propose a novel model, called \mdcoreall, to represent a bursting community in a temporal graph. Specifically, an \mdcoreall is a temporal subgraph in which each node has average degree no less than $\delta$ in a time segment with length no less than $l$. To compute the \mdcoreall, we first develop a novel dynamic programming algorithm which can calculate the segment density efficiently. Then, we propose an improved algorithm with several novel pruning techniques to further improve the efficiency. In addition, we also develop an efficient algorithm to enumerate all \mdcorealls that are not dominated by the others in terms of the parameters $l$ and $\delta$. The results of extensive experiments on 9 real-life datasets demonstrate the effectiveness, efficiency and scalability of our algorithms.
\end{abstract}

\section{Introduction} \label{sec:introduction}
Real-world networks such as social networks, biological networks, and communication networks are highly dynamic in nature. These networks can be modeled as graphs, and the edges in these graphs often evolve over time. In these graphs, each edge can be represented as a triple $(u,v,t)$, where $u,v$ are two end nodes of the edge and $t$ denotes the interaction time between $u$ and $v$. The graphs that involve temporal information are typically termed as temporal graphs~\cite{nature05,12temporalnetworksurvey}.

The interaction patterns in a temporal graph are often known to be bursty, e.g., the human communication events occur in a short time~\cite{nature05,12temporalnetworksurvey}. Here, a bursty pattern denotes a number of events occurring in a short time. In this paper, we study a particular bursty pattern on temporal networks, called bursting community, which denotes a dense subgraph pattern that occurs in a short time. In other words, we aim to identify densely-connected subgraphs from a temporal graph that emerges in a short time. Mining bursting communities from a temporal network could be useful for many practical applications, two of which are listed as follows.

\stitle{Activity discovery.} There are evidences that the timing of many human activities, ranging from communication to entertainment and work patterns, follow non-Poisson statistics, characterized by bursts of rapidly occurring events separated by long periods of inactivity~\cite{nature05}. For example, the talking points in temporal social networks such as Twitter, Facebook and Weibo are changing over time. By mining the bursting communities in such temporal social networks, we are able to identify a group of users that densely interact with each other in a short time. The common topics discussed among the users in a bursting community may represent an emerging activity that recently spreads over the networks. Therefore, identifying bursting communities may be useful for finding such emerging activities in a temporal network.

\stitle{Emergency event detection.} In communication networks (e.g., phone-call networks),  the users' communication behaviors may also exhibit bursty patterns. Identify bursting communities in a communication network may be useful for detecting emergency events. For instance, consider a scenario when an earthquake occurs in a country~\cite{journal2016MyShake}. Individuals in that country may contact their relatives and friends in a short time. These communication behaviors result in that many densely-connected subgraphs may be formed in a short time, which are corresponding to bursting communities. Therefore, by identifying bursting communities in a communication network  (e.g., a phone-call network) could be useful for detecting the emergency events (e.g., earthquake).

In the literature, there exist a few studies that are proposed to mine communities in temporal graphs.  For example, Wu et al. \cite{15bigdatatemporalcore} proposed a temporal $k$-core model to find cohesive subgraphs in a temporal graph. Ma et al.\cite{17icdedensegraphtemporal} devised a dense subgraph mining algorithm to identify densest subgraphs in a weighted temporal graph.
Rozenshtein et al.\cite{18icdmSegmentation} studied a problem of mining dense subgraphs at different time in a temporal graph. Li et al.\ \cite{18persistent} proposed an algorithm to find communities on temporal graphs that are persistent over time. Qin et al.\ \cite{19ICDEperiodicclique} studied a problem of finding periodic community in temporal networks. All of these studies do not consider the bursting patterns of the community, thus their techniques cannot be applied to solve our problem. 
To the best of our knowledge, we are the first to study the bursting community mining problem, i.e. the problem of finding the highly connected temporal subgraph in which each node is \textit{bursting out} in a short time. 

\stitle{Contributions.} In this paper, we formulate and provide efficient solutions to find bursting communities in a temporal graph. In particular, we make the following main contributions.

\vspace*{0.1cm}
\noindent \underline{\kw{Novel} \kw{model}.} We propose a novel concept, called \mdcoreall, to characterize the bursting community in temporal graphs. Each node in \mdcoreall has average degree no less than $\delta$ in a time segment with length no less than $l$. We also define a new concept called pareto-optimal \mdcoreall, which denotes the set of \mdcorealls that are not dominated by the other \mdcorealls in terms of the parameters $l$ and $\delta$. The pareto-optimal \mdcorealls can provide a good summary of all the bursting communities in a temporal graph over the entire parameter space.

\vspace*{0.1cm}
\noindent \underline{\kw{New} \kw{algorithms}.}
To find an \mdcoreall, the main technical challenge is to check whether a node $u$ has average degree no less than $\delta$ in a time segment with length no less than $l$. We show that the naive algorithm to solve this issue requires $O(|\mathcal{T}|^2)$ time, where $|\mathcal{T}|$ is the number of timestamps in the temporal network. To improve the efficiency, we first propose a dynamic programming algorithm which takes $O(|\mathcal{T}|)$ to solve this issue. Then, we develop a more efficient algorithm based on several in-depth observations of our problem which can achieve a near constant time complexity. In addition, we also propose an efficient algorithm to find the pareto-optimal \mdcorealls.

%

\vspace*{0.1cm}
\noindent \underline{\kw{Extensive} \kw{experimental} \kw{results}.}
We conduct comprehensive experiments using 9 real-life temporal graphs to evaluate the proposed algorithm. The results indicate that our algorithms significantly outperform the baselines in terms of the community quality. We also perform a case study on the \enron dataset. The results demonstrate that our approach can identify many meaningful and interesting bursting communities that cannot be found by the other methods. In addition, we also evaluate the efficiency of the proposed algorithms, and the results demonstrate the high efficiency of our algorithms. For example, on a large-scale temporal graph with more than 1M nodes and 10M edges, our algorithm can find a bursting community in 26.95 seconds. For reproducibility purpose, the source code of this paper is released at \url{https://github.com/VeryLargeGraph/MDC}. 

\stitle{Organization.} Section II introduces the model and formulates our problem. The algorithms to efficiently mining bursting communities are proposed in section III and IV. Experimental studies are presented in Section V, and the related work is discussed in Section VI. Section VII draws the conclusion of this paper.

\section{Preliminaries}
Let $\mathcal{G} = (\mathcal{V},\mathcal{E})$ be an undirected temporal graph, where $\mathcal{V}$ and $\mathcal{E}$ denote the set of nodes and the set of temporal edges respectively. Let $n = |\mathcal{V}|$ and $m = |\mathcal{E}|$ be the number of nodes and temporal edges respectively. Each temporal edge $e \in \mathcal{E}$ is a triplet $(u, v, t)$, where $u, v$ are nodes in $\mathcal{V}$, and $t$ is the interaction time between $u$ and $v$. Let ${\cal T}=\{t|(u, v, t) \in {\cal E}\}$ be the set of all timestamps. We assume without loss of generality that all the timestamps are sorted in a chronological order and they are joined as an arithmetic time sequence, i.e., $t_1 < t_2 < \cdots < t_{|{\cal T}|}$ and $(t_{i}-t_{i-1})$ is a constant of time interval for each integer $i$. In the rest of this paper, we use timestamps $\{0,1,2..,|\mathcal{T}|\}$ to represent $\{t_0,t_1,t_2..t_{|\mathcal{T}|}\}$. We assume that each timestamp is an integer, because the \textit{UNIX} timestamps are integers in practice.

For a temporal graph $\mathcal G$,  the \emph{de-temporal graph} of $\cal G$ denoted by $G=(V, E)$ is a graph that ignores all the timestamps associated with the temporal edges. More formally, for the de-temporal graph $G$ of $\cal G$, we have $V={\mathcal V}$ and $E=\{(u, v)| (u, v, t) \in {\cal E}\}$. Let $N_u(G)=\{v|(u, v) \in E\}$ be the set of neighbor nodes of $u$, and $deg_G[u]=|N_u(G)|$ be the degree of $u$ in $G$. For a given set of nodes $S \subseteq V$, a subgraph $G_S=(V_S, E_S)$ is referred to as an induced subgraph of $G$ from $S$ if $V_S = S$ and $E_S=\{(u, v)|u,v\in V_S, (u, v)\in E\}$.

Given a temporal graph $\cal G$, we can extract a series of \emph{snapshots} based on the timestamps. For each  $i \in {\cal T}$, we can obtain a snapshot $G_i=(V_i, E_i)$ where $V_i=\{u|(u, v, i) \in {\cal E}\}$ and $E_i=\{(u, v)|(u, v, i) \in {\cal E}\}$.  Fig.\ref{fig:toy-example} (a) illustrates a temporal graph $\cal G$ with 42 temporal edges and $\mathcal{T}=[1:6]$. Figs.\ref{fig:toy-example} (b) and (c) illustrates the de-temporal graph $G$ and all the six snapshots of $\cal G$ respectively.

\begin{figure}[t!] \vspace*{-0.5cm}
	\centering
	\subfigure[{\scriptsize Temporal edges in $\mathcal{G}$}]{
		\includegraphics[height=2.4cm]{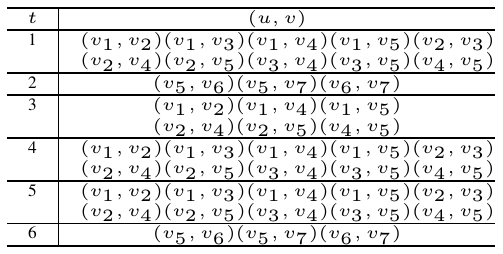}
	}
	\subfigure[{\scriptsize The de-temporal graph $G$}]{
		\includegraphics[height=2.1cm]{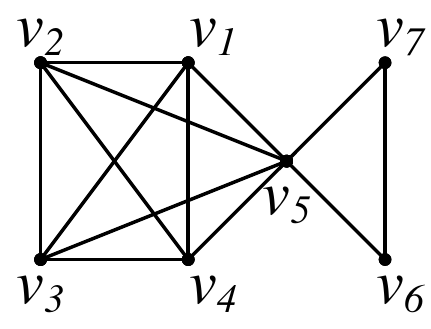}
	}
	\subfigure[{\scriptsize The six snapshots of $\mathcal{G}$}]{
		\includegraphics[height=2.04cm]{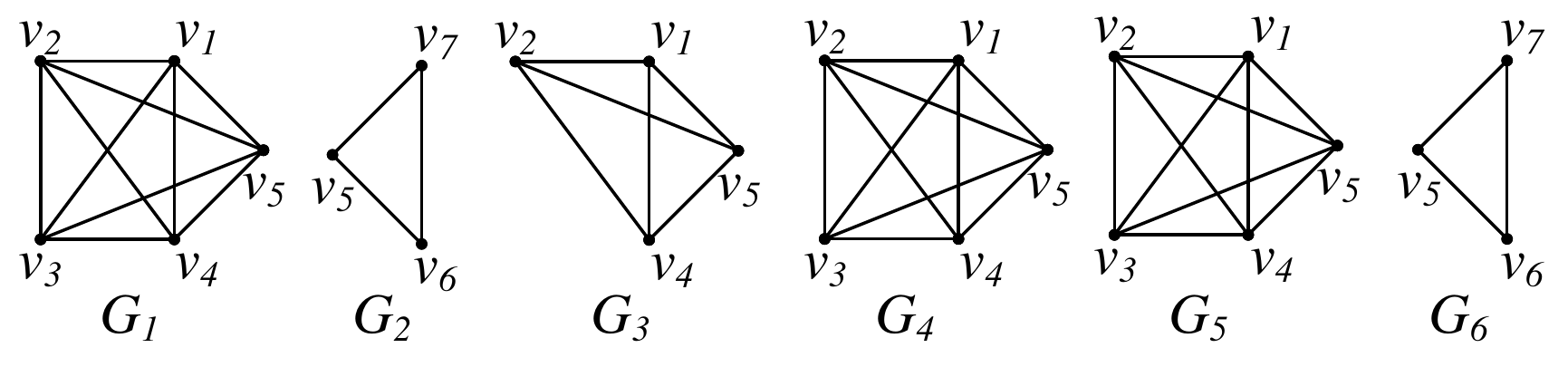}
	}
	\vspace*{-0.4cm}
	\caption{Basic concepts of the temporal graph}
	\label{fig:toy-example} \vspace*{-0.5cm}
\end{figure}

The nodes in bursting communities have a feature in common that they have high degrees in the induced subgraphs of some continuous time periods. We introduce some definitions below to describe the properties.

\vspace*{0.1cm}
\begin{definition}[temporal subgraph]\label{def:1subtemporal}
	Given a temporal graph $\mathcal{G} = (\mathcal{V},\mathcal{E})$, a continuous time interval $T = [t_s,t_e]\subseteq [1\text{:}|{\mathcal T}|]$ and a given set of nodes $S \subseteq \mathcal{V}$, a temporal subgraph can be denoted by $\mathcal{G}_S(T) = (S,\mathcal{E_S(T)})$, and it is an induced temporal graph of $\mathcal{G}$ from temporal edges $\mathcal{E}_S(T) = \{(u,v,t) | u,v \in S, t\in T, (u,v,t)\in \mathcal{E}\}$.
\end{definition}
\vspace*{0.1cm}

 Based on Definition~\ref{def:1subtemporal}, a temporal subgraph $\mathcal{G}_S(T)$ is an induced graph from nodes set $S$ in time interval $T$, it can extract a series of snapshots. 
 The snapshot of temporal subgraph in time $i$ is the induced subgraph of $V_i\cap S$, thus it can be denoted by $G_{V_i\cap S}$.
 For each node $u \in S$, $deg_{G_{V_i\cap S}}[u] = |N_u(G_{V_i\cap S})| = |N_u(G_i) \cap S|$.

\vspace*{0.1cm}
\begin{definition}[degree sequence]\label{def:2degseq}
	Given a temporal graph $\mathcal{G}_S(T)$, for node $u \in S$, the degree sequence of $u$ in $\mathcal{G}_S(T)$, abbreviated as $\kw{DS}(u, \mathcal{G}_S(T))$, is a sequence of $u$'s degree in each snapshot of $\mathcal{G}_S(T)$. Each item in the degree sequence can be denoted by $\kw{DS}(u, \mathcal{G}_S(T))[i] =|N_u(G_i) \cap S|$.
\end{definition}
\vspace*{0.1cm}


\vspace*{0.1cm}
\begin{definition}[$l$-segment density]\label{def:3segden}
	Given an integer $l$, a time interval $T=[t_s,t_e]$ and $u$'s degree sequence $\kw{DS}(u, \mathcal{G}_S(T))$, the $l$-segment density of $u$ in this degree sequence is the average degree of $u$ in $\kw{DS}(u, \mathcal{G}_S(T))$ while the length of the segment is no less than $l$, which can be denoted by
	
	$\lsdnospace(u, \mathcal{G}_S(T)) = \frac{\sum_{i=t_s}^{t_e} |N_u(G_i) \cap S| } {t_e-t_s+1}$ satisfying $t_e-t_s+1\ge l$
\end{definition}
\vspace*{0.1cm}


Based on Definition~\ref{def:3segden}, the \lmsdall of $u$ in $\mathcal{G}_S$ (abbreviated as \lmsdnospace$(u,\mathcal{G}_S)$), is the $l$-segment density $\lsdnospace(u, \mathcal{G}_S(T))$ such that there do not exist $S'\in \mathcal{V},T'\in \mathcal{T}$ satisfying $\lsdnospace(u, \mathcal{G}_{S'}(T'))>\lsdnospace(u, \mathcal{G}_S(T))$.


Below, we give a definition to describe the node which has average degree no less than $\delta$ in a time segment with length no less than $l$ in a given temporal subgraph.
\vspace*{0.1cm}
\begin{definition}[\densenode]\label{def:4densenode}
	Given a temporal graph $\mathcal{G}$, an integer $l$ and a real value $\delta$, one node $u$ is an \densenode in $\mathcal{G}$ if \lmsdnospace$(u,\mathcal{G})\ge \delta$.
\end{definition}
\vspace*{0.1cm}

According to Definition~\ref{def:4densenode}, we introduce a structure which can cluster the \densenodes.

\vspace*{0.1cm}
\begin{definition}[\mdcoreall] \label{def:4densecore}
	Given a temporal graph $\mathcal G$, an integer $l\ge 2$ and a real value $\delta>0$, an \mdcoreall (abbreviated as \mdcore) is a temporal subgraph $\mathcal{G}_C$ in which $C\subseteq\mathcal{V}$, satisfying\\
	$(i)$ \textit{Densely:} each node is an \densenode in $\mathcal{G}_C$, which means that $\forall u\in C$, \lmsdnospace$(u,\mathcal{G}_C)\ge \delta$ holds. \\
	$(ii)$ \textit{Maximally:} there does not exist a subset of nodes $S$ in $\mathcal{G}$ that satisfies $(i), (ii)$ and $C\subset S$.
\end{definition}
\vspace*{0.1cm}

Below, we use an example to illustrate the above definitions.

\begin{example}
	Consider the temporal graph in Fig.~\ref{fig:toy-example}. Given $l=3, \delta =3$. As shown in Fig.~\ref{fig:toy-example}(c), we can easily get that $\kw{DS}(v_5,\mathcal{G}) = [4,2,3,4,4,2]$. As $l=3$, then the \lmsdall \lmsdnospace$(v_5,\mathcal{G})= (3+4+4)/3 = 3.66$. Given $S =\{v_1,v_2,v_3,v_4,v_5 \}$, we can get that $\kw{DS}(v_5,\mathcal{G}_S) = [4,0,3,4,4,0]$, \lmsdnospace$(v_5,\mathcal{G}_S)= (3+4+4)/3 = 3.66$. Therefore, $v_5$ is a $(3,3)$-dense node in $\mathcal{G}_S$. Considering $v_3$ in $S$, we can get that $\kw{DS}(v_3,\mathcal{G}_S) = [4,0,0,4,4,0]$, \lmsdnospace$(v_3,\mathcal{G}_S)= (0+4+4)/3 =  2.66$. So, $v_3$ is a not $(3,3)$-dense node in $\mathcal{G}_S$. Therefore, $\mathcal{G}_S$ is not a $(3,3)$-\mdc. However, given $C =\{v_1,v_2,v_4,v_5 \}$, we can find that all the nodes in $C$ are $(3,3)$-dense nodes, because all the nodes have the \lmsdall of 3 considering $T=[3:5]$. So, $\mathcal{G}_C$ is a $(3,3)$-\mdc with $C = \{v_1,v_2,v_4,v_5 \}$.
	\done
\end{example}

\stitle{Problem 1 (Bursting Community).} Given a temporal graph $\mathcal{G}$, an integer $l\ge 2$ and a real value $\delta>0$, the goal of mining one bursting community is to compute the \mdcore in $\cal G$.

Based on Definition~\ref{def:4densecore}, \mdcore is a bursting community which can identify important events in the temporal graph, but it may be not easy to find proper parameters of $l$ and $\delta$ for practical applications. Intuitively, a good bursting community will have large $l$ and $\delta$ values. But large $l$ and $\delta$ values may result in losing answers. However, based on the theory of Pareto Optimality, we are able to compute the bursting communities that are not dominated by the other communities in terms of parameters $l$ and $\delta$. Below, we introduce a new concept, \skyline, to define those communities.

\vspace*{0.1cm}
\begin{definition}[\skyline] \label{def:5skyline}
	Given a temporal graph $\mathcal G$, an \mdcore in $\mathcal G$ is a \skyline if there does not exist a $(l',\delta')$-$\mdc$ in $\mathcal G$ such that $l'>l, \delta'\ge\delta $ or $l'\ge l, \delta'>\delta$.
\end{definition}
\vspace*{0.1cm}

Based on Definition~\ref{def:5skyline}, \skylines in the temporal graph $\mathcal{G}$ are summarizations of all the \mdcore. Intuitively, each \mdcore will be contained in one of the \skylines since they are maximal.

\stitle{Problem 2 (Pareto-optimal Bursting Community).} Given a temporal graph $\mathcal{G}$, the goal of mining Pareto-optimal bursting communities is to enumerate all the \skylines in $\cal G$.

\stitle{Hardness Discussion.} We can find that the problem of mining one bursting community is a little similar to mining traditional $k$-core. But it is not sufficient by adopting the traditional core decomposition method directly. One way to solve the problem is reducing the temporal graph by removing the nodes which are not \densenodes, and then checking whether the remained nodes are \densenodes until no nodes will be reduced. Therefore, many nodes will be checked whether are \densenodes in the remained graph again and again. The time complexity of the naive method to check whether one node is \densenode for one time is $O(|\mathcal{T}|^2)$. However, the status of one node must be checked while one edge is deleted, the times of the checking steps are $O(m)$. In some large temporal networks the scale of $|\mathcal{T}|$ is near to $m$, so the whole time complexity is near to $O(m^3)$. Clearly, this approach may involve numerous redundant computations for checking some nodes which are definitely not contained in an \mdcore.

To list all the \skylines, the naive method is to enumerate parameter pairs $(l,\delta)$ and outputs the one which can not be dominated. This way is difficult, because it is hard to set the proper $\delta$ which is a real value. However, another possible way is only considering one dimension, such as $l$ first, and then finding the maximal $\delta$. Next, we keep $\delta$ unchanged and find the maximal $l$. The challenge is how to acquire the answers with less redundant computations.

\section{Algorithms For Mining \mdcore }
In this section, we first introduce a basic decomposition framework to mine the \mdcore. Next, we develop a dynamic programming algorithm which can compute the segment density efficiently, and then propose an improved algorithm with several novel pruning techniques.

\subsection{The \kw{MDC} Algorithm}\label{sub:3.1}

We can observe that \mdcore has the following three properties.
\begin{property}[Uniqueness]
	Given parameters $l> 1$ and $\delta> 0$, the \mdcore of the temporal graph $\mathcal{G}$ is unique.
\end{property}

\begin{proof}
	We can prove this lemma by a contradiction. Suppose	that there exist two different \mdcorealls in $\mathcal{G}$, denoted by $C_1$ and	$C_2$ respectively $(C_1 \neq C_2)$. Let us consider the node set $C' = C_1 \cup C_2$. Following Definition~\ref{def:4densecore}, every node in $C'$ is a \densenode in $\mathcal{G}(C')$, because it is a \densenode in $\mathcal{G}_{C_1}\cup \mathcal{G}_{C_2}$. Since $C_1 \neq C_2$, we have $C_1\subset C'$ and $C_2\subset C'$ which contradicts to the fact that $C_1(C_2 )$ satisfies the maximal
	property.
\end{proof}

\begin{property}[Containment]
	Given an \mdcore of the temporal graph $\mathcal{G}$, the $(l,\delta')$-\mdc with $\delta'\ge \delta$ is a temporal subgraph of \mdcore.
\end{property}

\begin{proof}
	According to Definition~\ref{def:4densecore}, an \mdcore $C$ is a maximal temporal subgraph, and any node in $C$ has segment density at least $\delta$ with length no less than $l$. For $\delta'\ge \delta$, each node in $(l,\delta')$-maximal dense core will also have segment density at least $\delta$ with length no less than $l$. Since the $C$ is a maximal temporal subgraph, $(l,\delta')$-maximal dense core must be contained in $C$.
\end{proof}

We first give the definition of $k$-CORE, and then show the third property.
\textbf{The $k$-CORE of the de-temporal graph of $\mathcal{G}$ can be denoted by $G_c =(V_c,E_c)$}, which is a maximal subgraph such that $\forall u \in G_c: deg_{G_c}[u] \ge k$.

\begin{property}[Reduction] \label{pro:03}
	Given an \mdcore of the temporal graph $\mathcal{G}$, the nodes in \mdcore must be contained in the $k$-CORE $(k=\delta)$ of the de-temporal graph $G$.
\end{property}

\begin{proof}
	According to Definition~\ref{def:4densecore}, any node $u$ in an \mdcore $\mathcal{G}_C$ has segment density at least $\delta$ with length no less than $l$ $(l\ge 2)$. So, $u$ must have degree at least $\delta$ in at least one snapshot $G^*$. As each $G^* \subseteq G$, each $u$ in $C$ must have degree no less than $\delta$. Since the $k$-CORE $(k=\delta)$ of the de-temporal graph $G$ is the maximal subgraph such that each nodes have degree no less than $\delta$, $C$ must be contained in the $k$-CORE $(k=\delta)$ of $G$.
\end{proof}

Following the property 3.3, we first compute the $k$-CORE $(k=\delta)$ of the de-temporal graph of $\mathcal{G}$, denoted by $G_c$. Given the properties of \textit{Uniqueness} and \textit{Containment}, we can apply the core decomposition framework to compute the \mdcore. Next, we check whether or not node $u$ satisfies the \textit{Densely} property mentioned in Definition~\ref{def:4densecore}. Specifically, we compute the $G_c$ in $G$ first, and then check whether node $u$ is an \densenode for all $u\in G_c$. If $u$ is not an \densenode, we delete $u$ from the results. Since the deletion of $u$ may result in $u$'s neighbors no longer being the \densenode, we need to iteratively process $u$'s neighbors. The process terminates if no node can be deleted. The details are provided in Algorithm~\ref{alg:1densecore}.

\begin{algorithm}[t!]\vspace*{-0.5mm}
	\scriptsize
	\caption{$\mdc({\mathcal G}, l, \delta)$ }
	\label{alg:1densecore}
	\KwIn{Temporal graph $\mathcal{G} = (\mathcal{V},\mathcal{E})$, parameters $l$ and $\delta$}
	\KwOut{\mdcore in $\mathcal{G}$}
	
	Let $G=(V, E)$ be the de-temporal graph of $ {\mathcal G}$\;
	Let $G_c=(V_c, E_c)$ be the $k$-CORE $(k=\delta)$ of $G$\;
	${\cal Q} \gets [\emptyset]; D \gets [\emptyset]; \mathcal{MSD}\gets [\emptyset]$\;
	\For {$u \in V_c$} {
		$deg[u] \gets |N_u(G_c) |$;  /* compute the degree of $u$ in $G_c$*/ \\
		$\mathcal{MSD}[u] \gets \kw{ComputeMSD}({\mathcal G}, l, u, V_c)$\;
		{\bf if }{$\mathcal{MSD}[u] < \delta$}{ \bf then}{
			${\cal Q}.push(u)$\;
		}
	}
	
	\While{${\cal Q} \neq [\emptyset]$}{
		$v \gets {\cal Q}.pop(); \ D \gets D \cup \{ v\}$\;
		\For{$w\in N_v(G_c)$, s.t. $deg[w]\ge \delta$ and $\mathcal{MSD}(w)\ge \delta$}{
			$deg[w] \gets deg[w]- 1$\;
			{\bf if}{ $deg[w] < \delta$}{ \bf then}{
				${\cal Q}.push(w)$\;
			}
			\Else{
				$\mathcal{MSD}[w] \gets \kw{ComputeMSD}({\mathcal G}, l, w, V_c \setminus D)$\;
				{\bf if }{$\mathcal{MSD}[w] < \delta$}{ \bf then}{
					${\cal Q}.push(w)$\;
				}
			}
		}
	}
	{\bf return} $\mathcal{G}_{V_c \setminus D}$ \;
	
\end{algorithm}

Algorithm~\ref{alg:1densecore} first computes the $k$-CORE $(k=\delta)$ of the de-temporal graph ${G}$ (lines 1-2), denoted by $G_c= (V_c,E_c)$. Then, it initializes a queue $\mathcal{Q}$ to store the nodes to be deleted, a set $D$ to store the deleted node, a collection $\mathcal{MSD}$ to store \lmsdall for each node (line 3) and $deg[u]$ to store the degree of $u$ in $G_c$ (line 5). Next, for each $u$ in $V_c$, it invokes Algorithm~\ref{alg:2ComputeDensity} to check whether $u$ is an \densenode or not (lines 4-6). If $u$'s \lmsdall $\mathcal{MSD}[u]$ is less than $\delta$, $u$ is not an \densenode and it will be pushed into a queue $\mathcal{Q}$ (lines 7-8). Subsequently, the algorithm iteratively processes the nodes in $\mathcal{Q}$. In each iteration, the algorithm pops a node $v$ from $\mathcal{Q}$ and uses $D$ to maintain all the deleted nodes (line 10).
For each neighbor node $w$ of $v$, the algorithm updates $deg[w]$ (lines 12). If the revised $deg[w]$ is smaller than $\delta$, $w$ is clearly not an \densenode.
As a consequence, the algorithm pushes $w$ into $\mathcal{Q}$ which will be deleted in the next iterations (line 13). Otherwise, the algorithm invokes Algorithm 2 to determine whether $w$ is an \densenode (lines 14-15). The algorithm terminates when $\mathcal{Q}$ is empty. At this moment, the remaining nodes $V_c \setminus D$ is the \densenodes of $\mathcal{G}$, and the algorithm returns temporal subgraph $\mathcal{G}_{V_c\setminus D}$ (line 16).

\begin{example}
	Recall the temporal graph in Fig.~\ref{fig:toy-example}. Given $l=3, \delta =3$. Algorithm~\ref{alg:1densecore} first computes the $k$-CORE $(k=\delta)$ of de-temporal graph $G$. So, $V_c = \{ v_1,v_2,v_3, v_4,v_5\}$. Then, for each node $u$ in $V_c$, it checks whether $u$ is an \densenode. Consider $v_3$, $\kw{DS}(v_3, \mathcal{G}_{V_c}) = [4,0,0,4,4,0]$, we can not find a segment of at least 3 length in which the density is no less than 3. Next, $v_3$ will be pushed into $\mathcal{Q}$. In line 9, $v_3$ is added into set $D$ and all of its neighbors will be checked in line 10. Now the remained nodes are $ \{v_1,v_2, v_4,v_5\}$, and we can find that the $deg$ and $\mathcal{MSD}$ of them are no less than 3. Therefore, Algorithm~\ref{alg:1densecore} returns $\mathcal{G}_{V_c\setminus D}$ with $V_c\setminus D= \{v_1,v_2, v_4,v_5\}$.
\end{example}

\stitle{Correctness of Algorithm~\ref{alg:1densecore}.} Let $C=V_c\setminus D$. It will check $\mathcal{MSD}[u]$ and call procedure \kw{ComputeMSD} once its neighbor is deleted and added into $D$,	so each node in the remained $C$ must have a $l$-segment density at least $\delta$ in $\mathcal{G}_C$ with length no less than $l$. Therefore, each node in the remained $C$ will have the same property. According to Definition~\ref{def:4densecore}, Algorithm~\ref{alg:1densecore} correctly computes \mdcore. \done

\stitle{Complexity of Algorithm~\ref{alg:1densecore}. } The time and space complexity of Algorithm 1 by invoking Algorithm~\ref{alg:2ComputeDensity} to compute \msd is $O(m|\mathcal{T}|)$ and $O(m)$ respectively.

\begin{proof}
	First, Algorithm~\ref{alg:1densecore} needs $O(m)$ time to compute the $k$-CORE in the de-temporal graph $G$ (line 2). As Algorithm~\ref{alg:2ComputeDensity} needs time of $O(|\mathcal{T}|)$ (see Section~\ref{sub:3.2}), it takes $O(|\mathcal{T}| n)$ time to initialize queue $\mathcal{Q}$ for all the nodes in $V_c$ (lines 4-7). Next, in lines 8-16, for each node $v$, the algorithm explores all neighbors of $v$ at most once. So it will invoke Algorithm~\ref{alg:2ComputeDensity} at most $m$ times and the total time complexity is $O(|\mathcal{T}| m)$. Therefore, since $n<m$, the total time complexity of Algorithm~\ref{alg:1densecore} is $O(|\mathcal{T}| m)$.
	
	In this algorithm, we need to maintain the graph and store collections of $\mathcal{Q},D$ and $deg$ which consumes $O(m)$ in total. In procedure \kw{ComputeMSD}, it needs space of $O(|\mathcal{T}|)$ (see Section~\ref{sub:3.2}). Since $|\mathcal{T}| <m$, thus the total space complexity of Algorithm~\ref{alg:1densecore} is $O(m)$.
\end{proof}

Different from the traditional core decomposition algorithm, Algorithm~\ref{alg:1densecore} needs to check whether one node is an \densenode in each iteration. Below, the implementation details of \kw{ComputeMSD} are described.

\subsection{Dynamic Programming Procedure of \kw{ComputeMSD}} \label{sub:3.2}
Recall Definition~\ref{def:4densecore}, one node $u$ is an \densenode if \lmsdnospace$(u,\mathcal{G}_C)\ge \delta$ in the temporal subgraph $\mathcal{G}_C$. Considering one node $u$, we can get $u$'s degree sequence $\kw{DS}(u, \mathcal{G}_C)$ inside the candidate \mdcore for $i$ range from $1$ to $|\mathcal{T}|$ first, and then compute the \lmsdall of $u$.
For convenience, in this subsection we denote $\kw{DS}(u, \mathcal{G}_C)$ by $\mathcal{DS}[u] = \{ |N_u(G_i) \cap C|, i \in [1 :|\mathcal{T}|] \}$, \lmsdnospace$(u,\mathcal{G}_C)$ by $\mathcal{MSD}[u]$ while they all consider the degree sequence in $\mathcal{G}_C$.
To get $\mathcal{MSD}[u]$, the naive method is to considering all the segment of longer than $l$, but the time complexity is $O(|\mathcal{T}|^2)$. Below, we propose a dynamic programming algorithm which transforms the problem into finding the maximum slope in a curve, which can reduce the computational overhead to linear complexity.

\vspace*{0.1cm}
\begin{definition}[cumulative sum curve]
	\label{def:7csc}
	Given node  $u$'s degree sequence $\kw{DS}(u, \mathcal{G}_C)$ (abbreviated as $\mathcal{DS}[u]$), the cumulative sum curve (abbreviated as $\mathcal{CSC}$) of $u$ is a collection of $\{ \mathcal{CSC} [i] = \sum_{i=1}^t{\mathcal{DS}[u][i]} , t \in [1 :|\mathcal{T}|] \}$.

\end{definition}
\vspace*{0.1cm}

Without loss of generality, we set $\mathcal{CSC}[0]$ as 0. Then, the points $\{(0,\mathcal{CSC}[0]),(1,\mathcal{CSC}[1])... (|\mathcal{T}|,\mathcal{CSC}[|\mathcal{T}|]) \}$ can be drawn as a curve in the Cartesian Coordinate System, and we denote this curve by $\mathcal{CSC}$.  Next, we define the slope by considering two points in $\mathcal{CSC}$.

\vspace*{0.1cm}
\begin{definition}[slope]\label{def:8slope}
	Given integers $i,j\in [1,|\mathcal{T}|],i<j$, the slope of curve $\mathcal{CSC}$ from $i$ to $j$ can be denoted by $\kw{slope}(i,j,\mathcal{CSC})=\frac{\mathcal{CSC}[j] - \mathcal{CSC}[i-1]}  {j-i+1}$, where $i,j$ can be marked as the $start$ and $end$ of the slope, respectively.
\end{definition}
\vspace*{0.1cm}

For convenience, we abbreviate $\kw{slope}(i,j,\mathcal{CSC})$ as $\kw{slope}(i,j)$ in the following paper while the symbol $\mathcal{CSC}$ can not be confused.

\begin{lemma}
	\label{lemma:slop}
	For a degree sequence $\mathcal{DS}[u]$, one time interval $[t_s,t_e]$, the segment density of the subsequence in $[t_s,t_e]$ equals the slope of curve $\mathcal{CSC}$ from $t_e$ to $t_s$. Formally,
	$\frac{\sum_{i=t_s}^{t_e} \mathcal{DS}[u][i]}  {t_e-t_s+1}  = \kw{slope}(t_s,t_e,\mathcal{CSC})$.
\end{lemma}

\begin{proof}
	The proof can be easily obtained by the definitions, thus we omit it for brevity.
\end{proof}

\vspace*{0.1cm}
\begin{definition}[maximum $j$-truncated $l$-slope]
	\label{def:mts}
	Given a curve $\mathcal{CSC}$ of node $u$, a truncated time $j\in[l:|\mathcal{T}|]$, the maximum $j$-truncated $l$-slope $\mtsnospace[j] = \{\max(\kw{slope}(i,j)) \| i = [0,j-l]\}$.
\end{definition}
\vspace*{0.1cm}

According to Lemma~\ref{lemma:slop} and Definition~\ref{def:mts}, $\mtsnospace[j]$ is the maximum slope which ended at time $j$ and the length of the corresponding segment is no less than $l$. For convenience, \mts is the collection of \{$\mtsnospace[j], j\in [l, |\mathcal{T}|] $\}.

\begin{corollary}
	\label{coro:slop}
	The problem of finding the \lmsdnospace$(u,\mathcal{G}_C)$, can be transformed to computing $\max(\mtsnospace)$ in $\mathcal{CSC}$ of $u$.
	
\end{corollary}

\begin{proof}
	According to lemma~\ref{lemma:slop}, the problem of finding the \lmsdall, can be transformed to computing the maximum slope of the curve $\mathcal{CSC}$ in which the difference between the $start$ and $end$ of slope is no less than $l$.
		
	However, there exists the maximum slope of the curve $\mathcal{CSC}$ which ended at some time $t'$.
	If we range $t$ from time $l$ to time $|\mathcal{T}|$ and record all the $\mtsnospace[t]$, then the maximum one will be the maximum slope of the curve. According to Definition~\ref{def:mts}, the difference between the $start$ and $end$ of $\mtsnospace[t]$ is at least $l$. Therefore, \lmsdnospace$(v,\mathcal{G}_C) = \max(\mtsnospace)$.
\end{proof}

Next, the problem is how to compute all the $\mtsnospace[t]$ with $t=[1: |\mathcal{T}|]$. One efficient idea is maintaining $\mtsnospace[t+1]$ by the computed $\mtsnospace[t]$ and the changes of the curve from time $t$ to ${t+1}$. Below, considering the computed $\mtsnospace[t]$ and the newly joined point $(t+1,\mathcal{CSC}[{t+1}])$, we can maintain \mts based on the following observations.

\begin{observation}
	We can compute a lower convex hull (abbreviated as \ch) in $\mathcal{CSC}$ of $u$ which ended at time $t-l$, the slope of the tangent from point $(t,\mathcal{CSC}[t])$ to the \ch is the \lmsdall of node $u$ ended at time $t$.
\end{observation}

\begin{observation} If the point $(a,\mathcal{CSC}[a])$ and $(b,\mathcal{CSC}[b])$ is on the maintained lower convex hull, suppose that $a<b<c$, \ch will add node $(c,\mathcal{CSC}[c])$ and remove node $(b,\mathcal{CSC}[b])$ if $(\mathcal{CSC}[c]-\mathcal{CSC}[b])/(c-b) \le (\mathcal{CSC}[b]-\mathcal{CSC}[a])/(b-a)$.
\end{observation}

\begin{observation} For one ended time $t$, if $(\mathcal{CSC}[t]-\mathcal{CSC}[b])/(t-b) \ge (\mathcal{CSC}[b]-\mathcal{CSC}[a])/(b-a)$, then the slope of $\mathcal{CSC}[t]$ to $\mathcal{CSC}[a]$ will not be the maximum one and node $(a,\mathcal{CSC}[a])$ should be removed from $\mathcal{CH}$.
\end{observation}

Following the observations above, we devise an algorithm to maintain the lower convex hull \ch ended at time $t-l$, and the $\mtsnospace[t]$ can be computed in a recursive way as the following algorithm shows.

\begin{algorithm}[t!]\vspace*{-0.5mm}
	\scriptsize
	\caption{$\kw{ComputeMSD}({\mathcal G}, l, u, C)$ }
	\label{alg:2ComputeDensity}
	$\mathcal{CSC}\gets [\emptyset];\mathcal{CSC}[0]\gets 0 ; \mathcal{DS}[u] \gets [\emptyset]$ \;
	\For {$t \gets 1 : {|\mathcal T|}$} {
		Let $G_t$ be the snapshot of $\cal G$ at timestamp $t$\;
		$\mathcal{DS}[u][t] \gets |N_u(G_t) \cap C|$\;
		$\mathcal{CSC}[t]$ = $\mathcal{CSC}[t-1] + \mathcal{DS}[u][t]$\;
		
	}
	\ch$\gets [\emptyset], i_s \gets 0 , i_e \gets -1, \mathcal{MTS}[u] \gets  [\emptyset]$\;
	\For{$t \gets l  : {|\mathcal T|}$}  {
		\While{$i_s < i_e$ and $\kw{slope}(\chnospace[i_e],t-l,\mathcal{CSC})$ $\leq$ $\kw{slope}(\chnospace[i_e-1],\chnospace[i_e],\mathcal{CSC})$
		} {
			$i_e\gets i_e -1$\;
		}
		$\chnospace[++i_e] \gets t - l$;
		
		\While{$i_s < i_e$ and $\kw{slope}(\chnospace[i_s],t,\mathcal{CSC})$ $\geq$ $\kw{slope}(\chnospace[i_s],\chnospace[i_s+1],\mathcal{CSC})$
		} {
			$i_s \gets i_s +1$\;
		}
		
		${\mathcal{MTS}[u]} \gets {\mathcal{MTS}[u]} \cup \{\kw{slope}(\chnospace[i_s],t,\mathcal{CSC}) \} $\;
	}
	
	{\bf return} $\max(\mathcal{MTS}[u])$\;
	
	\vspace{2mm}
	
	{\bf Procedure} $\kw{slope}(i, j ,\mathcal{CSC})$ \\
	{\bf return}  $(\mathcal{CSC}[j] - \mathcal{CSC}[i])/(j-i)$
	
\end{algorithm}

Algorithm~\ref{alg:2ComputeDensity} first initializes $\mathcal{CSC}[t]$ of $u$ for all timestamps (lines 1-5). As the nodes set $C$ may be changed in Algorithm~\ref{alg:1densecore}, the degree of $u$ can be computed in line 4. Next, it maintains an array \ch to record the indexes of each points in the lower convex hull, $i_s$ to record the $start$ index of \ch, $i_e$ to record the $end$ index of \ch and $\mathcal{MTS}[u]$ to record \mts (line 6). For time $t$ from $l$ to $|\mathcal{T}|$, it dynamically computes $\mtsnospace[t]$ of $u$ (lines 7-13).  In lines 8-9, $i_e$ reduces by 1 if the $\kw{slope}(\chnospace[i_e], t-l)$ is no larger than $\kw{slope}(\chnospace[i_e-1], \chnospace[i_e])$, because the rear node point will be above the convex hull \ch by the end of $t-l$ following Observation 3.2. If there is no such point in the end, $\chnospace[++i_e]$ is assigned by $t-l$. In lines 11-12, the head index adds up by 1 if $\kw{slope}(\chnospace[i_s], t)$ is no larger than $\kw{slope}(\chnospace[i_s], \chnospace[i_s+1])$, because it will have an upper convex hull in the curve of \ch at the start of $\chnospace[i_s]$ according to Observation 3.3.
We will get a array $\mathcal{MTS}[u]$ of $\mtsnospace[t]$ with $t$ ranges from $l$ to $|\mathcal{T}|$. Finally, it returns $\max(\mathcal{MTS}[u])$ after all the iterations (line 14).

\begin{figure}[t]\vspace*{-0.5cm}
	\centering
	\subfigure[$t=4$]{
		\includegraphics[width=2.51cm]{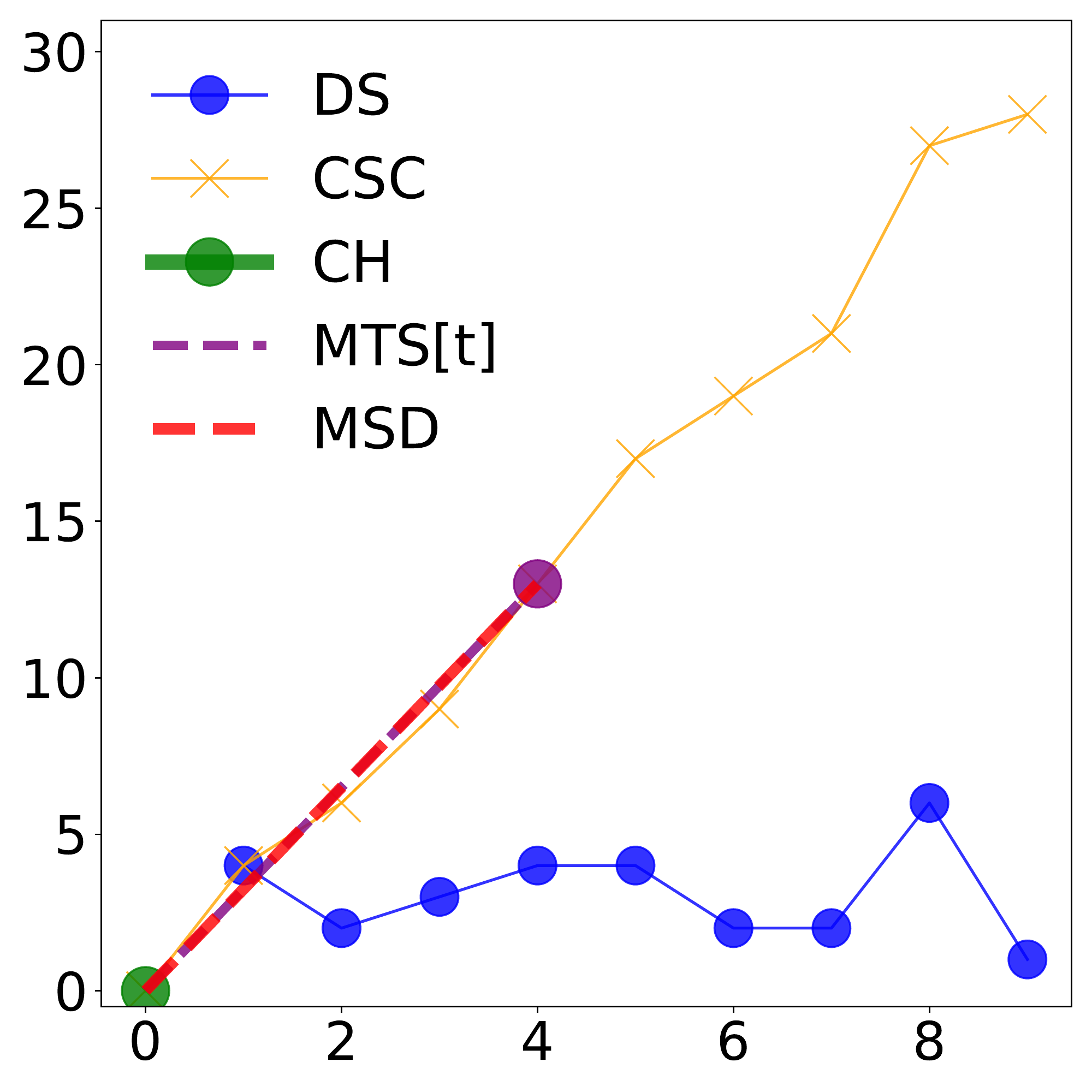}
	}
	\subfigure[$t=5$]{
		\includegraphics[width=2.51cm]{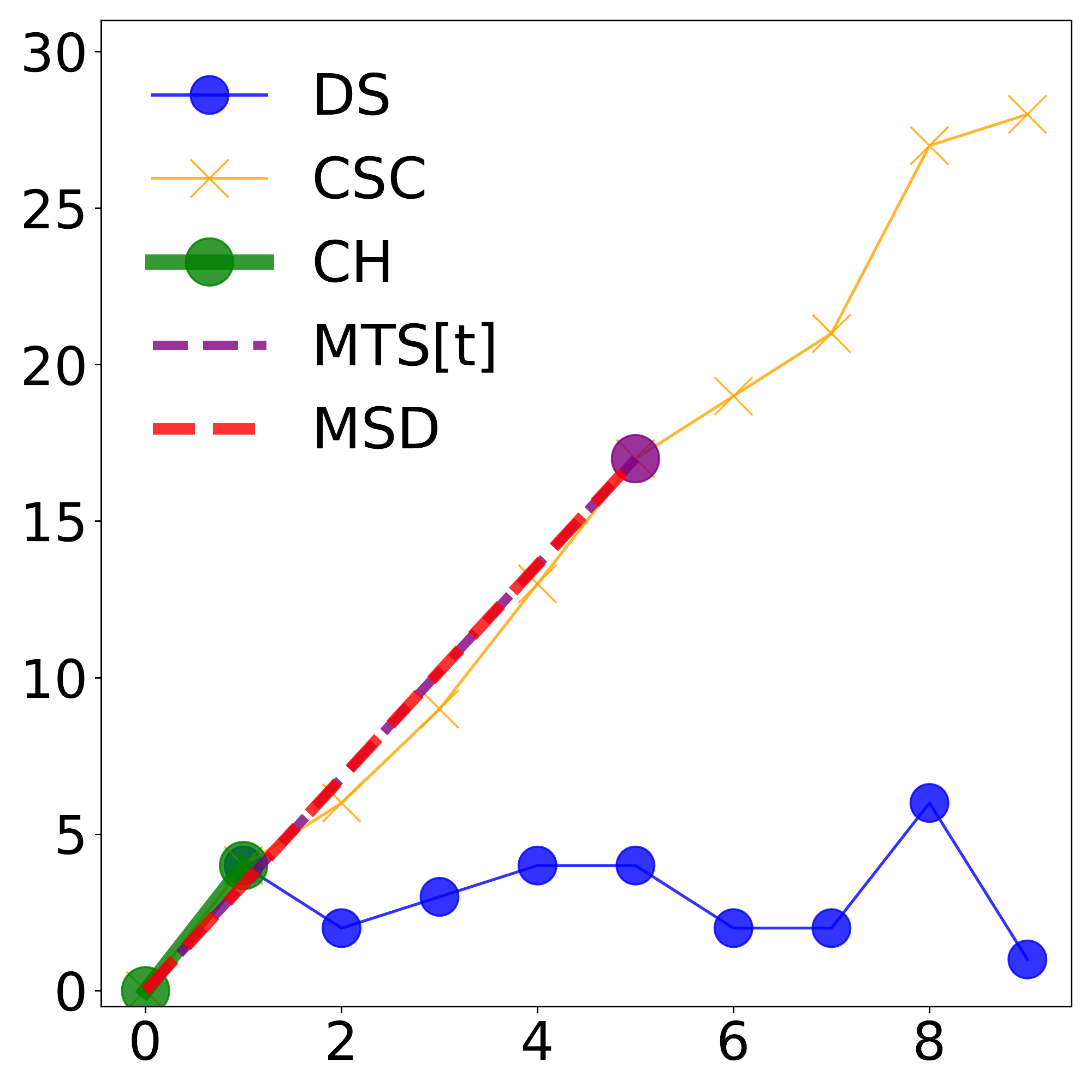}
	}
	\subfigure[$t=6$]{
		\includegraphics[width=2.51cm]{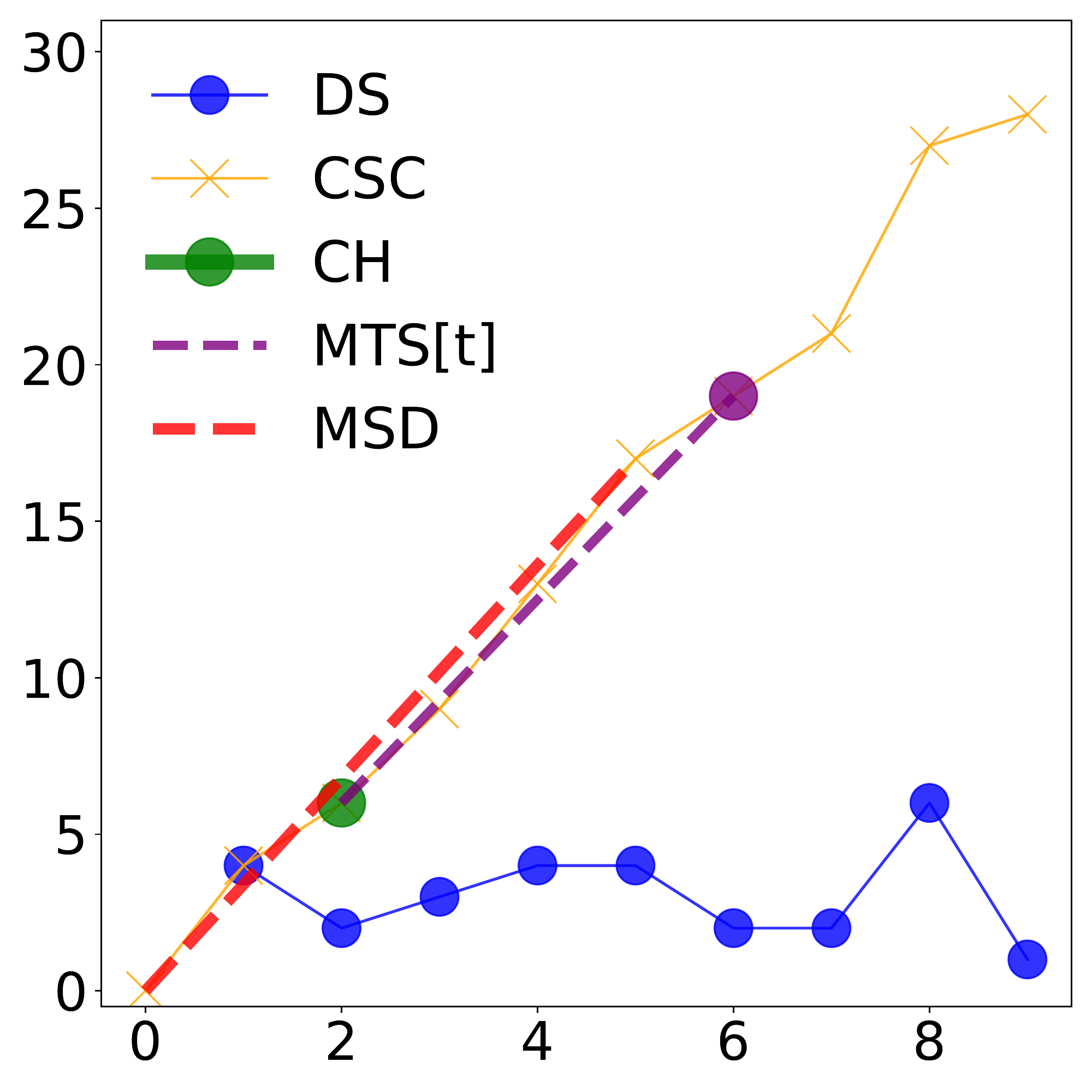}
	}
	\subfigure[$t=7$]{
		\includegraphics[width=2.51cm]{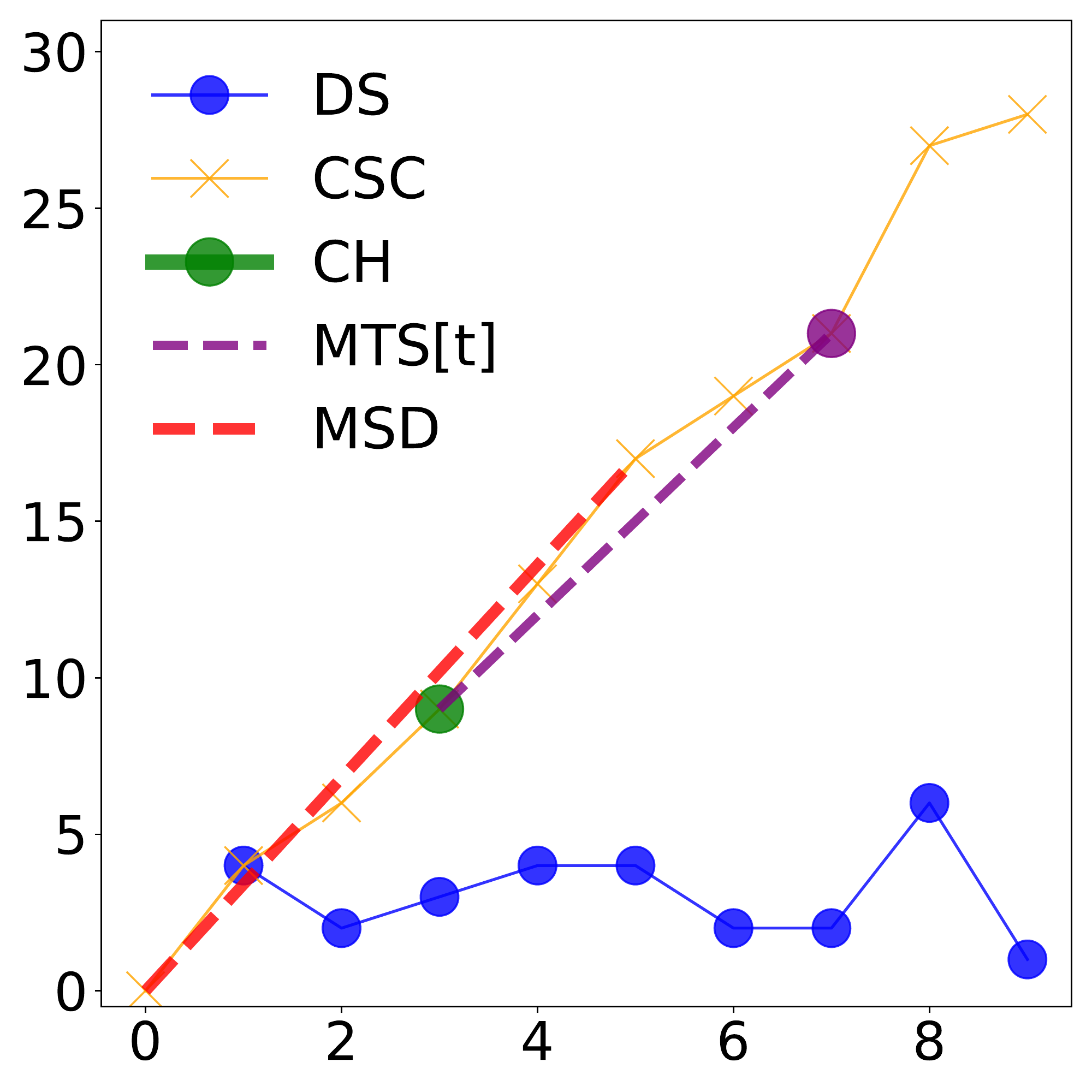}
	}
	\subfigure[$t=8$]{
		\includegraphics[width=2.51cm]{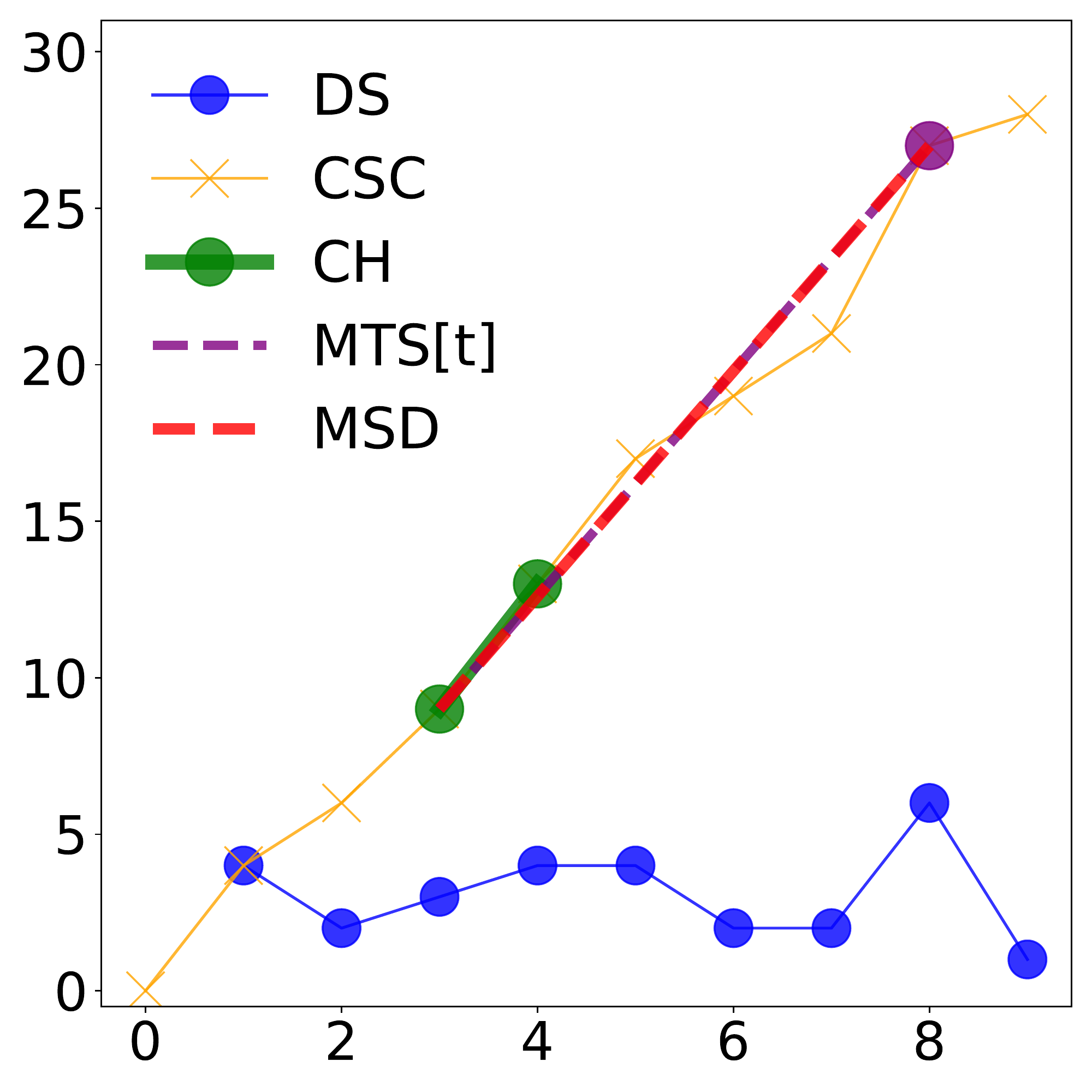}
	}
	\subfigure[$t=9$]{
		\includegraphics[width=2.51cm]{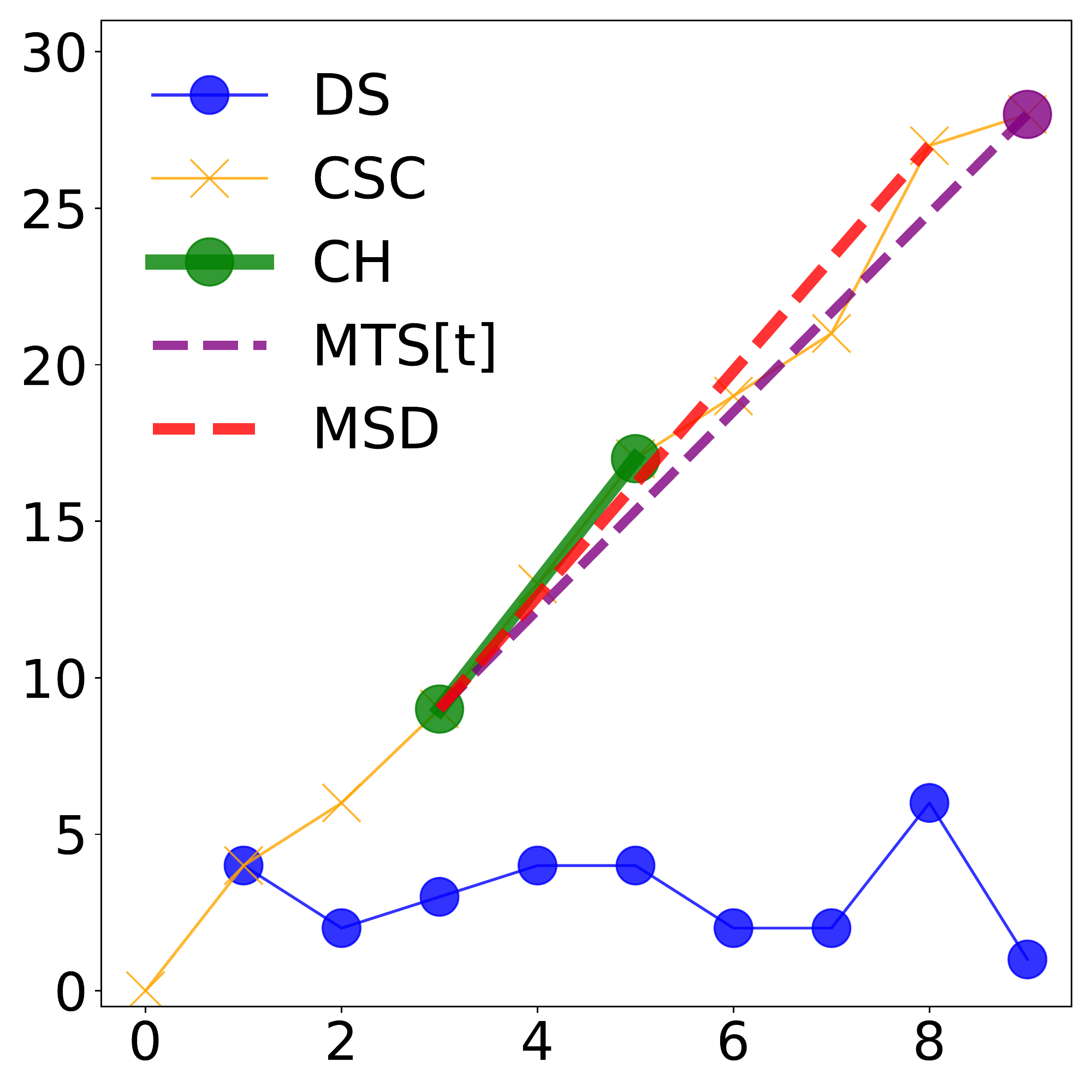}
	}
	\vspace*{-0.2cm}\caption{Running example of computing \lmsdall for a degree sequence of $[4, 2, 3, 4, 4, 2, 2, 6, 1]$ with $l=4$}
	\vspace*{-0.5cm}
	\label{mdc-running}
\end{figure}

\begin{example}
	Fig.~\ref{mdc-running} shows the running example of computing \lmsdall for a degree sequence of $[4, 2, 3, 4, 4, 2, 2, 6, 1]$ with $l=4$. Clearly, $\mathcal{T} =[1 : 9]$, $\mathcal{CSC} = [0, 4, 6, 9, 13, 17, 19, 21, 27, 28]$. According to Corollary~\ref{coro:slop}, the procedure starts at $t=4$ because we need satisfy that the length of the segment is no less than $l$. At this time, there is only one item in \ch. When $t=5$, the $i_e$ index of \ch adds up by $1$ (line 10), but the $i_s$ index is remained $0$ because $\kw{slope}(0,5) = (17-0)/(5-0) = 3.4$ is no larger than $\kw{slope}(0,1) = (4-0)/(1-0) =4.0$ (lines 12). And $\max(\mathcal{MTS}[u])$ is currently $\mtsnospace[5] =3.4$. Next, $t=6$, according to Observation 3.2, the $i_e$ index of \ch reduces by $1$ because $\kw{slope}(1,2) = 2.0$ is no larger than $\kw{slope}(0,1) =4.0$ (lines 8-9). Then, the newly $i_e$ is $1$ and $\chnospace[i_e]$ is assigned by $t-l=2$ (line 10). Now \ch is $[0,2], i_s=0, i_e=1$. In the next step, the $i_s$ index adds up by $1$ because $\kw{slope}(0,6) = 19/6>\kw{slope}(0,2)=6/2$ (line 12). So, the final \ch and $\mtsnospace[6]$ can be shown at Fig.~\ref{mdc-running}(c). Likewise, when $t=7 \ to \ 9$, the \ch will be maintained by the similar processes. It should be noted that when $t=7$, $\kw{slope}(3,8) = 3.6$, which is larger than $\mtsnospace[5]$. Finally, $\mathcal{MSD} = \max(\mathcal{MTS}[u]) =3.6$, which is the density of the $4th$ to $8th$ items $[4,4,2 ,2, 6]$.
	\done
\end{example}

\stitle{Correctness of Algorithm~\ref{alg:2ComputeDensity}. } According to Corollary~\ref{coro:slop}, we need to prove that $(i)$ $\max(\mathcal{MTS}[u])$ is the maximum slope; $(ii)$ the length of corresponding segment is no less than $l$.
For $(i)$, according to Observation $3.3$, lines 11-13 will compute $\mtsnospace[t]$ which will be recorded in $\mathcal{MTS}[u]$, thus the final $\max(\mathcal{MTS}[u])$ is the maximum slope.
For $(ii)$, $t-\chnospace[i_s]\ge l$ because the only assignment code for $\chnospace[i]$ is in line 10, and $t$ will be larger in the next loop, so $t-\chnospace[i] \ge l$ for any $i$. \done

\stitle{Complexity of Algorithm~\ref{alg:2ComputeDensity}. } For a temporal graph $\mathcal{G}$ with $|\mathcal{T}|$ timestamps, 
the time and space complexity of Algorithm~\ref{alg:2ComputeDensity} is $O(|\mathcal{T}|)$ and $O(|\mathcal{T}|)$ respectively.


\begin{proof}
	First, Algorithm~\ref{alg:2ComputeDensity} needs $O(|\mathcal{T}|)$ to compute the collection $\mathcal{CSC}$ (lines 2-5). For each $t$, $i_e$ reduces from $t$ to $i_s$ (lines 8-9), and $i_s$ increases from $l$ to $|\mathcal{T}|$ (lines 11-12). Considering all the loops, the average time complexity of assigning $i_s$ is $O(|\mathcal{T}|)$. Since $i_e$ has the lower bound of $i_s$ in each loop, the average time of assigning $i_e$ is also $O(|\mathcal{T}|)$. Hence, the whole Algorithm~\ref{alg:2ComputeDensity} need $O(|\mathcal{T}|)$ to calculate node $u$'s \lmsdall.
	
	In this algorithm, we store collections of $\mathcal{CSC},\mathcal{CH}$ and $\mathcal{MTS}[u]$ which consume $O(|\mathcal{T}|)$ in total. Therefore, the space complexity of Algorithm~\ref{alg:2ComputeDensity} is $O(|\mathcal{T}|)$.
\end{proof}

\subsection{An improved \kw{MDC+} algorithm} \label{sub:3.3}
Although Algorithm~\ref{alg:1densecore} is efficient in practice, it still has two limitations. (i) It still needs to call the procedure \kw{ComputeMSD} for all nodes in $V_c$ (line 6 in Algorithm~\ref{alg:1densecore}). In the worst case, the time complexity of this process can be near to $|\mathcal{T}|m$. We can observe that if we delete a certain node $u$, the $deg[v]$ of $u$'s neighbor $v$ will reduce, and we can monitor it at once to check whether $deg[v]<\delta$. Once $deg[v]<\delta$, we do not need to call the procedure \kw{ComputeMSD} for $v$ any more. (ii) It still needs to compute all the \lmsdall dynamically for each deletion of the edges. We can observe that in each call of $\kw{ComputeMSD}$, the degree of $u$ reduces only one and $\mathcal{MSD}[u]$ may not change. So, the $\kw{ComputeMSD}$ algorithm clearly results in significant amounts of redundant computations for the iterations for all $t$ from $l$ to $ |\mathcal{T}|$.

To overcome this limitation, we propose an improved algorithm called $\kw{MDC+}$. The striking features of $\kw{MDC+}$ are twofold. On one hand, it needs not to call procedure $\kw{ComputeMSD}$ for each node in advance. Instead, it calculates $\kw{SD}$ of the candidate node on-demand. On the other hand, when deleting a node $u$, $\kw{MDC+}$ does not re-compute $\mathcal{MSD}[w]$ for a neighbor node $w$ of $u$. Instead, $\kw{MDC+}$ dynamically updates the computed $\mathcal{MSD}[w]$ for each node $w$, thus substantially avoiding redundant computations. The detailed description of $\kw{MDC+}$ is shown in Algorithm~\ref{alg:3MaintainDensity}.

\begin{algorithm}[t]\vspace*{-0.5mm}
	\scriptsize
	\caption{$\kw{MDC+}({\mathcal G}, l, \delta)$ }
	\label{alg:3MaintainDensity}
	\KwIn{Temporal graph $\mathcal{G} = (\mathcal{V},\mathcal{E})$, parameters $l$ and $k$}
	\KwOut{\mdcore in $\mathcal{G}$}
	Let $G=(V, E)$ be the de-temporal graph of $ {\mathcal G}$\;
	Let $G_c=(V_c, E_c)$ be the $k$-CORE $(k=\delta)$ of $G$\;
	Let $deg[u]$ be the degree of $u$ in $G_c$\;
	
	${\cal Q} \gets [\emptyset]; D \gets [\emptyset]; \mathcal{MSD} \gets [\emptyset]; \mathcal{MTS} \gets [\emptyset]; \mathcal{DS} \gets [\emptyset]$\;
	\For {$u \in V_c$ in an increasing order by $deg[u]$ } {
		{\bf if} $u \in D$ {\bf then}  {\bf continue}\;
		$(\mathcal{MTS}[u] , \mathcal{DS}[u]) \gets \kw{ComputeMSD}^*({\mathcal G}, l, u, V_c\setminus D)$;
		\qquad  \qquad  \qquad  \qquad   /* all the same to Alg.~\ref{alg:2ComputeDensity} except that it returns $(\mathcal{MTS}[u], \mathcal{DS}[u])$ */\\
		$\mathcal{MSD}[u] \gets \max(\mathcal{MTS}[u])$\;
		{\bf if} $\mathcal{MSD}[u] < \delta$ {\bf then}  \{$\mathcal{Q}.push(u); deg[u] \gets 0;$\}\\
		\While{${\cal Q} \neq \emptyset$}{
			$v \gets {\cal Q}.pop();  D \gets D \cup \{ v\}$\;
			\For{$w\in N_v(G_c)\setminus D $, s.t. $deg[w]\ge \delta$ }{
				$deg[w] \gets deg[w]- 1$\;
				{\bf if }{ $deg[w] < \delta$}{ \bf then} \{${\cal Q}.push(w)$; {\bf continue;}\}
				
				{\bf if} {$\mathcal{MSD}[w]$ is not existed}{ \bf then }{\bf continue}\;
				\For{$ t$, s.t.$(v,w,t) \in \mathcal{E}$}{
				$\mathcal{DS}[w][t] \gets \mathcal{DS}[w][t]-1$\;
				$\mathcal{MSD}[w] \gets \kw{UpdateMSD}(w, t, l,\mathcal{DS}, \mathcal{MTS})$\;
				
				}
				{\bf if }{$\mathcal{MSD}[w] < \delta$ }{\bf then}  \{$\mathcal{Q}.push(w); deg[w] \gets 0;$\}\\
			}
		}
	}
	{\bf return} $\mathcal{G}_{V_c \setminus D}$ \;
	
		\vspace{2mm}
	
	{\bf Procedure} $\kw{UpdateMSD}(w, t, l,\mathcal{DS}, \mathcal{MTS})$ \\
	$\mathcal{CSC}\gets [\emptyset]$; $t_s\gets \max(0,t-2l)$; $t_e\gets \min(t+2l, |\mathcal{T}|)$; $\mathcal{CSC}[0]\gets 0$\;
	\For {$i \gets 0: t_e-t_s$} {
		$\mathcal{CSC}[i+1]$ = $\mathcal{CSC}[i] + \mathcal{DS}[w][t_s+i]$\;	
	}
	\ch$\gets [\emptyset], i_s \gets 0 , i_e \gets -1$\;
	\For{$j \gets l: t_e-t_s+1$}  {	
		\While{$i_s < i_e$ and $\kw{slope}(\chnospace[i_e],j-l,\mathcal{CSC})$ $\leq$ $\kw{slope}(\chnospace[i_e-1],\chnospace[i_e],\mathcal{CSC})$
		} {
			$i_e\gets i_e -1$\;
		}
		$\chnospace[++i_e] \gets t - l$;
		
		\While{$i_s < i_e$ and $\kw{slope}(\chnospace[i_s],j,\mathcal{CSC})$ $\geq$ $\kw{slope}(\chnospace[i_s],\chnospace[i_s+1],\mathcal{CSC})$
		} {
			$i_s \gets i_s +1$\;
		}
		\If{$j\ge t-t_s$}{
			$\mathcal{MTS}[w][j+t_s-l]\gets \kw{slope}(\chnospace[i_s],j,\mathcal{CSC}) $\;
		}
	}
	{\bf return} $\max(\mathcal{MTS}[w])$\;
	
\end{algorithm}

Algorithm~\ref{alg:3MaintainDensity} first computes the $k$-CORE $(k=\delta)$ $G_c$ in the de-temporal graph (line 2).
Next, it explores the nodes in $V_c$ based on an increasing order by the degrees in $G_c$ (line 5). When processing a node $u$, the algorithm first checks whether $u$ has been deleted or not (line 6). If $u$ has not been removed, $\kw{MDC+}$ invokes Algorithm~\ref{alg:2ComputeDensity} to compute $\mathcal{MSD}[u]$ (lines 7-8). It should be noted that the procedure $\kw{ComputeMSD}^*$ is all the same to $\kw{ComputeMSD}$ except that it returns $(\mathcal{MTS}[w], \mathcal{DS}[u])$ (replace line 14 of Algorithm~\ref{alg:2ComputeDensity}).
Next, if $\mathcal{MSD}[u]$ is no larger than $\delta$, $u$ is not an \densenode. Thus, the algorithm pushes $u$ into the queue $\mathcal{Q}$ (line 9). Subsequently, the algorithm iteratively deletes the nodes in $\mathcal{Q}$ (lines 10-19). When removing a node $v$, $\kw{MDC+}$ explores all $v$'s neighbors (line 12). For a neighbor node $w$, \kw{MDC+} first updates the degree of $w$ (line 13), i.e., $deg[w]$. If the updated degree is less than $\delta$, $w$ is not an \densenode (line 14). In this case, the algorithm pushes it into $\mathcal{Q}$ and continues to process the next node in $\mathcal{Q}$ (the degree pruning rule).
Otherwise, if $\mathcal{MSD}[w]$ has already been computed, the algorithm invokes $\kw{UpdateMSD}$ to update $\mathcal{MSD}[w]$ (line 19). If the updated $\mathcal{MSD}[w]$ is less than $\delta$, $w$ is not an \densenode and the algorithm pushes $w$ into $\mathcal{Q}$ (line 19). We can see that if $\mathcal{MSD}[w]$ has not been computed yet, the algorithm does not need to update $\mathcal{MSD}[w]$. In this case, $\mathcal{MSD}(w)$ will be calculated in the next iterations of line 7. It also should be noted that the $\mathcal{DS}$ is always updated, because if the nodes have been deleted by the degree constraint, $\mathcal{DS}$ will be newest in $\mathcal{G}_{V_c \setminus D}$ (line 7), otherwise if the nodes have been deleted by the $(l,\delta)$-dense constraint, $\mathcal{DS}$ will be updated in line 17. Finally, \kw{MDC+} outputs $\mathcal{G}_{V_c \setminus D}$ as the result.

In the following, we introduce the \kw{UpdateMSD} procedure. Suppose that before updating, the \lmsdall of $w$ exists from time $t_s$ to $t_e$. At this time, if $\mathcal{DS}[w][t']$ reduces by $1$, then there exist three situations:
$(i) \ t' <t_s; (ii) \ t_s \le t' \le t_e; (iii) \ t'>t_e$.

\begin{figure}[t]\vspace*{-0.5cm}
	\centering
	\subfigure[$t'=1$]{
		\includegraphics[width=2.51cm]{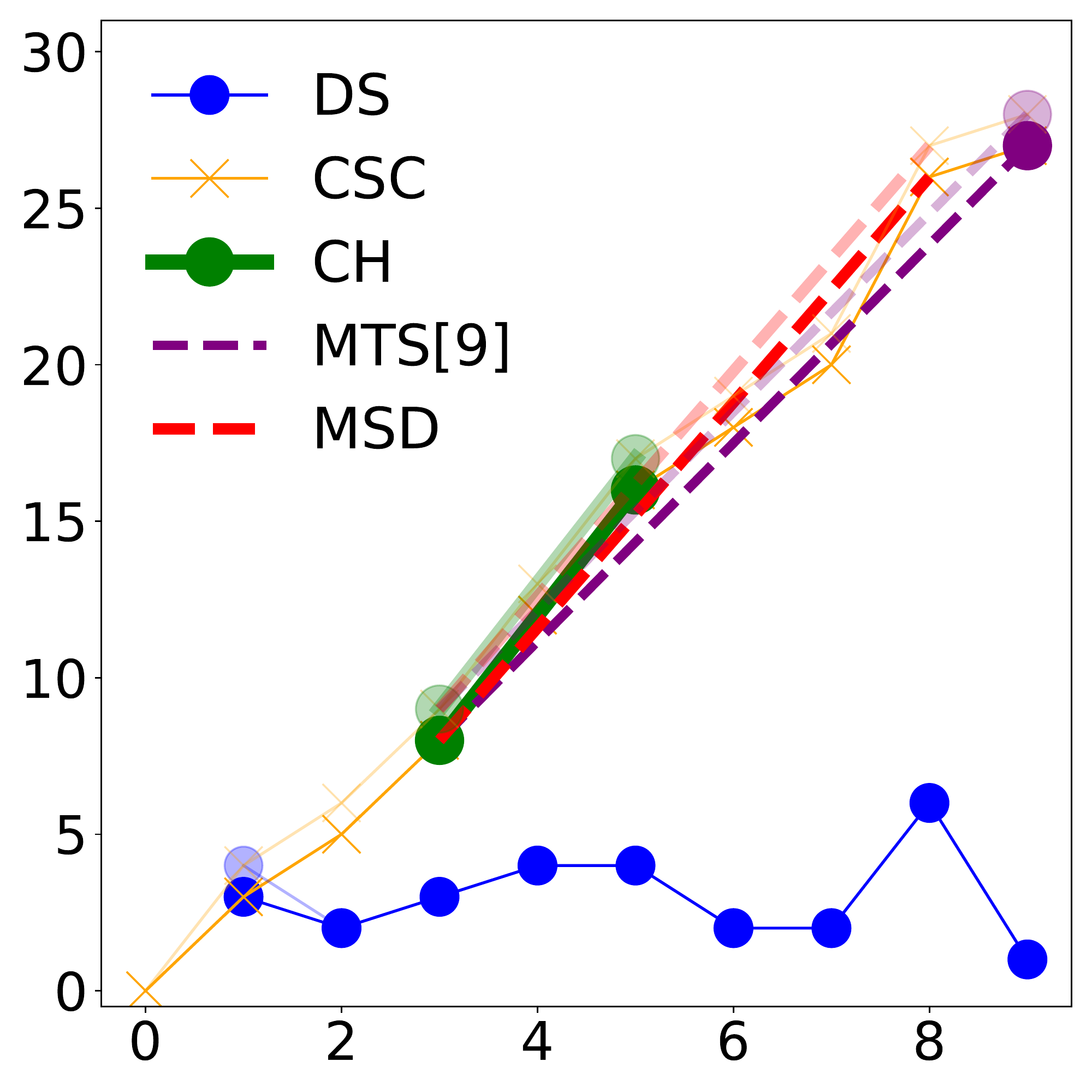}
	}
	\subfigure[$t'=4$]{
		\includegraphics[width=2.51cm]{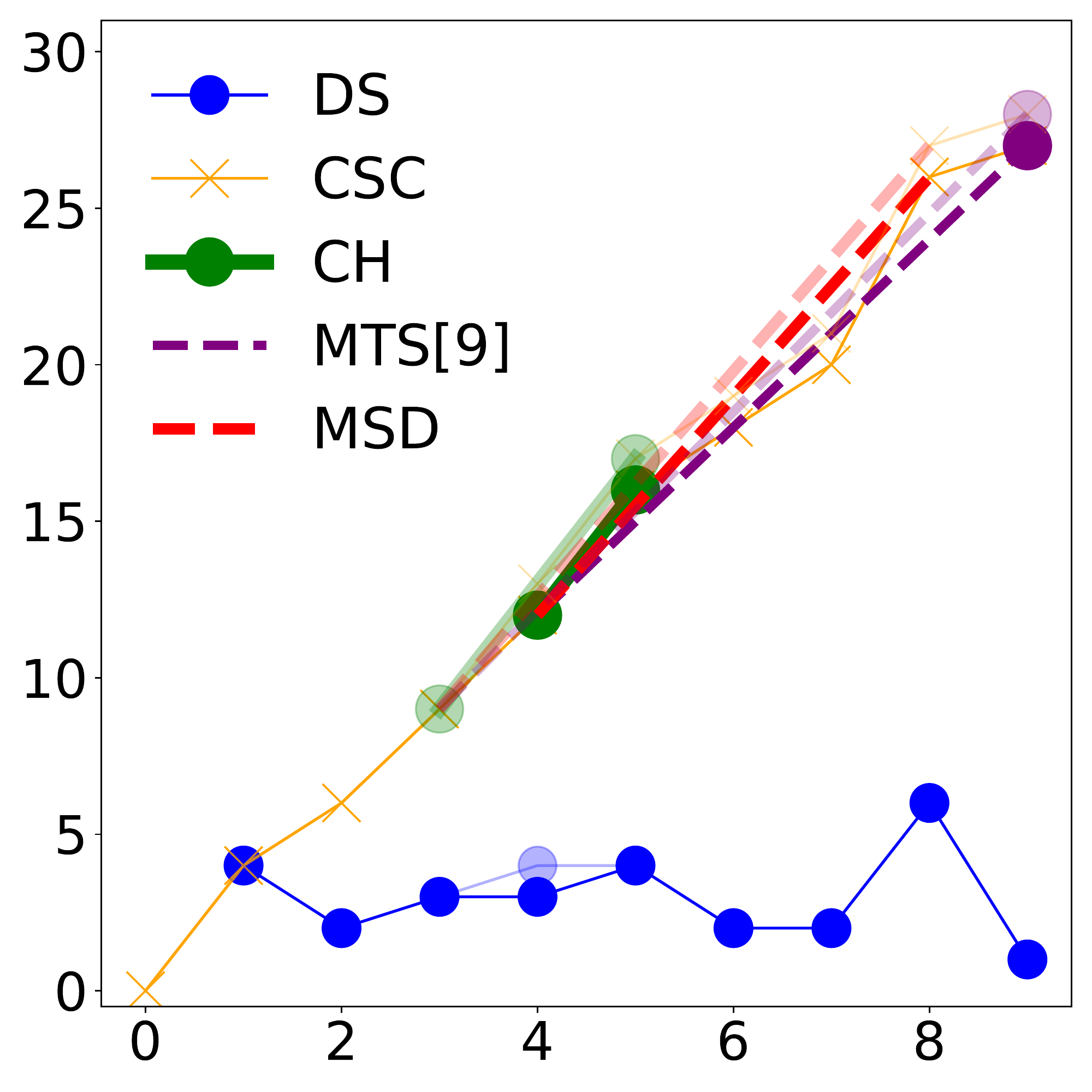}
	}
	\subfigure[$t'=9$]{
		\includegraphics[width=2.51cm]{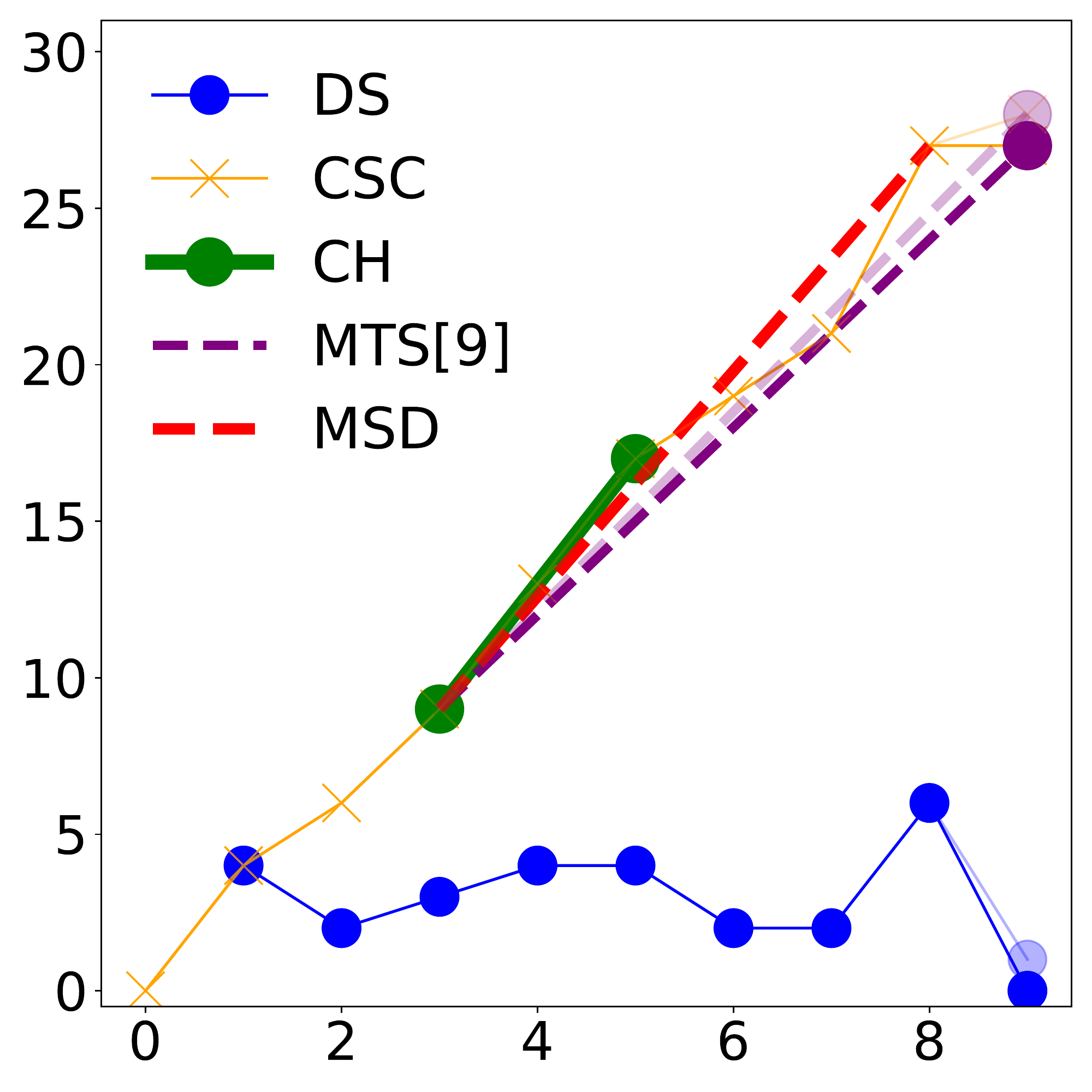}
	}
	\vspace*{-0.2cm}\caption{Updated situations of Fig.~\ref{mdc-running}(f) after $\mathcal{DS}[w][t']$ reduces by 1}
	\vspace*{-0.5cm}
	\label{fig:mdc+}
\end{figure}

\begin{example}
	Fig.~\ref{fig:mdc+} shows the three situations of Fig.~\ref{mdc-running}(f) after $\mathcal{DS}[w][t']$ reduces by 1.
	We can see that the current \lmsdall of $w$ exists from $t_s = 3$ to $t_e = 8$.
	As shown in Fig.~\ref{fig:mdc+}(a) in which $t'<t_s$ and Fig.~\ref{fig:mdc+}(c) in which $t'>t_s$, we can see that the $\mathcal{MSD}[w]$ will not change. We can find that the parts of curve with the maximum slop are all moved down. Also, it can be proved easily from the definition of $l$-segment density that $\mathcal{MSD}[w]$ will not change. Howerver, in Fig.~\ref{fig:mdc+}(b), $\mathcal{DS}[w][4]$ reduces by 1 and the new sequence is $[3, 2, 3, 3, 4, 2, 2, 6, 1]$. The \lmsdall is 3.5, which is the density of the $5th$ to $8th$ items $[4,2,2,6]$. So only when $t_s \le t' \le t_e$ should we update the $\mathcal{MSD}$.
	\done
\end{example}

Below we will introduce that it only needs to consider $\mathcal{DS}$ from time $t-2l$ to time $t+2l$ to update $\mathcal{MSD}$. We first define a concept, \mtsinospace$[j]$, which is a maximum $j$-truncated $l$-slope of considering only $2l$ length of the curve $\mathcal{CSC}$.

\vspace*{0.1cm}
\begin{definition}[maximum $j$-truncated $l$-slope of $2l$-length]
	\label{def:mts*}
	Given a curve $\mathcal{CSC}$ of node $u$ by Definition~\ref{def:7csc}, a truncated time $j\in[l:|\mathcal{T}|]$, the maximum $i$-lower $j$-truncated $l$-slope of $2l$-length
	\mtsinospace$[j] = \{\max(\kw{slope}(i,j)) \| i = [j-2l,j-l]\}$.
\end{definition}
\vspace*{0.1cm}

However, \mtsinospace$[j]$ is the maximum slope which only considers $\mtsnospace[j]$ with the slope ends at $j$ and starts in $[j-2l:j-l]$. Furthermore, it holds the property below.

\begin{lemma}
	\label{lemma:2l}
	Given a curve $\mathcal{CSC}$ of $u$, \lmsdnospace$(u,\mathcal{G}_C) = \max($\mtsinospace$)$ holds.
\end{lemma}

\begin{proof}
	If $l<t<2l$, then we can compute $\mtsnospace[t]$ by considering time from $l$ to $t$, which satisfies that $t-2l<l$, so $\mtsnospace[t]=$\mtsinospace$[t]$.
	If $t>2l$, suppose that the start time of $\mtsnospace[t]$ is $t^*$ and it holds $t^*<t-2l$  .  Since  $\mtsnospace[t]$ is maximum, there holds $\kw{slope}(t^*,t) \ge \kw{slope}(t-l,t)$. Thus, $\frac{\Sigma_{i=t^*}^{t}\kw{DS}[u][i]}{t-t^*+1}\ge \frac{\Sigma_{i=t-l+1}^{t}\kw{DS}[u][i]}{l}   \Rightarrow \frac{\Sigma_{i=t^*}^{t-l}\kw{DS}[u][i]}{t-l-t^*+1}\ge \frac{\Sigma_{i=t^*}^{t}\kw{DS}[u][i]}{t-t^*+1}  $.
	Based on Corollary~\ref{coro:slop}, if $t^*<t-2l$, then $ \frac{\Sigma_{i=t^*}^{t-l}\kw{DS}[u][i]}{t-l-t^*+1}  \le\mtsnospace[t-l]$. Thus, we can have the result that if $t^*<t-2l$, then $\mtsnospace[t] \le \mtsnospace[t-1]$. Therefore, ended at time $t$,  $\mtsnospace[t-l]$ will be the most possible final $\max(\mtsnospace)$.
	In conclude, we can check the maximal $l$-slope \mtsinospace$[t]$ which ends at time $t$ and starts from $t-2l$ to $t-l$. Then, \mtsinospace$[t]$ with $t \in [l :|\mathcal{T}|]$ can be denoted by \mtsinospace, which satisfies that $\max($\mtsinospace$) = \max(\mtsnospace)$.
\end{proof}

\begin{corollary}
	\label{coro:2l}
	Given stored \mts of node $u$, we can have \lmsdnospace$(u,\mathcal{G}_C) = \max(\mtsnospace)$. If $\mathcal{DS}[u]$ reduces by 1 at time $t$, we only need to update $\mtsnospace[t']= $ \mtsinospace$[t']$ with $t'\in [t : \ t+2l]$  to get the updated \lmsdnospace$(u,\mathcal{G}_C)$.
	
\end{corollary}
	
\begin{proof}
	Suppose that after $\mathcal{DS}[u]$ reduces, the exactly maximum $j$-truncated $l$-slope and the one with $2l$-length are $\mts^*$ and $\mtsistar$, respectively. If \mts is updated by \mtsinospace$[t']$ with $t'\in [t : t+2l]$, the new set $\overline{\mts} = [\mts[0:t-1]:\mtsistar[t:t+2l]: \mts[t+2l+1:|\mathcal{T}|]]$. Suppose that the index of maximum one in $\mts^*$ is $t^*$,  there hold three situations: $(i) \ t^*<t$, then $\mts[0:t-1] = \mts^*[0:t-1]$, so $\max(\overline{\mts}) = \max(\mts^*)$; $(ii) \ t<t^*<t+2l$, \mtsinospace$[t:t+2l] = \mtsistar[t:t+2l]$ so $\max(\overline{\mts}) = \max(\mtsistar)$;
	$(iii) \ t^*>t+2l$, as we can have $\mtsistar[i] \le \mts^*[i]$ for each $i$, then $ \max(\mts^*[t:|\mathcal{T}|]) = \max([\mtsistar[t:t+2l]: \mts[t+2l+1:|\mathcal{T}|] ])$ so $\max(\overline{\mts}) = \max(\mts^*)$. According to Lemma~\ref{lemma:2l}, we can have $\max(\mts^*) = \max(\mtsistar) = \max(\overline{\mts})$.
\end{proof}

\begin{corollary}
	\label{coro:-2l}
	 If $\mathcal{DS}[u]$ reduces by 1 at time $t$, we only need to use $\mathcal{DS}[u][t']$ with $t' \in [\max(0,t-2l): \min(t+2l, |\mathcal{T}|)]$ to update $\mathcal{MTS}[u]$ and get the updated \lmsdnospace$(u,\mathcal{G}_C)$.
\end{corollary}

\begin{proof}
	According to Corollary~\ref{coro:2l}, to get the updated \lmsdnospace$(u,\mathcal{G}_C)$, we only need to update $\mtsnospace[t']= $ \mtsinospace$[t']$ with $t'\in [t : \ t+2l]$. Based on Definition~\ref{def:mts*}, we need build $\mathcal{CSC}$ in $[t-2l:t]$ to compute $\mtsnospace[t]$. In conclude, we only need to use $\mathcal{DS}[u][t']$ with $t' \in [\max(0,t-2l): \min(t+2l, |\mathcal{T}|)]$ to get the updated \lmsdnospace$(u,\mathcal{G}_C)$.
\end{proof}

According to the above corollaries, the $\kw{UpdateMSD}$ procedure first initializes $t_s$ as the left side of the considered time interval, $t_e$ as the right side and $\mathcal{CSC}$ based on Definition~\ref{def:7csc} (lines 22-24). The following step is aimed at computing all the \mtsinospace$[j]$ which ends at time $j$ and starts from time $j-2l$ to $j-l$. The following process is much same as that in Algorithm~\ref{alg:2ComputeDensity} (lines 27-31). Note that, we use $\mathcal{MTS}[w][j]$ to record \mtsinospace$[j]$ of node $w$ and it should be updated only when $j \ge t-t_s$ (line 32). After all the $\mathcal{MTS}[w][j]$ with $j$ from $t$ to $t_s$ have been maintained, the procedure returns $\max(\mathcal{MTS}[w])$ as the updated $\mathcal{MSD}[w]$ (line 34).

\stitle{Correctness of Algorithm~\ref{alg:3MaintainDensity}. } We need to prove that $(i)$ $\mathcal{DS}$ is correctly updated; $(ii)$ the updated $\max(\mathcal{MTS})$ is always the maximum slope; $(iii)$ all the remained nodes in $V_c \setminus D$ has a \lmsdall no less than $\delta$.
For $(i)$, in line 7, $\mathcal{DS}$ is computed by considering the current remained nodes $V_c\setminus D$, thus it is the current exact one; in line 17, we can find that $\mathcal{DS}$ is updated for each deletion of temporal edges, unless the considering node $w$ has been popped into $\mathcal{Q}$. For $(ii)$, based on Corollary~\ref{coro:-2l}, in lines 16-18, each deletion of temporal edge $(v,w,t)$ has been considered so $\max(\mathcal{MTS})$ is always the exact answer. For $(iii)$, we can see that each node need to be checked (line 5) whether to have a \lmsdall (line 9) unless it has been deleted (line 6), thus the returned $\mathcal{G}_{V_c \setminus D}$ must be \mdc.
\done

\begin{lemma}
	\label{lemma:updatemsd}
	For a temporal graph G with $|\mathcal{T}|$ timestamps, procedure $\kw{UpdateMSD}$ need $O(l)$ to maintain the \lmsdall.
\end{lemma}

\begin{proof}
First, $\kw{UpdateMSD}$ needs $O(l)$ to compute the collection $\mathcal{CSC}$ (lines 23-24). For each $j$, $i_e$ reduces from $j$ to $i_s$, and $i_s$ increases from $t_s$ to $t_e$. Considering all the loops, the average time complexity of assigning $i_s$ and $i_e$ is $O(t_e-t_s) = O(l)$ (lines 26-33). And computing $\max(\mathcal{MTS}[w])$ needs $O( l)$ because we can use $\mathcal{MTS}[w][t_s:t_e]$ with the former maximum value to compute it. However, the whole procedure $\kw{UpdateMSD}$ needs $O(l)$ to update $\mathcal{MSD}[w]$.
\end{proof}

\stitle{Complexity of Algorithm~\ref{alg:3MaintainDensity}. } The time and space complexity of Algorithm 3 are $O(\alpha |\mathcal{T}|+\beta l)$ and $O(\alpha |\mathcal{T}|+m)$ respectively, where $\alpha = |V_c|, \beta = |E_c|$ are number of nodes and edges in $k$-CORE $(k=\delta)$  of $G$.

\begin{proof}
	First, Algorithm~\ref{alg:3MaintainDensity} needs $O(m)$ time to compute the $k$-CORE $G_c = (V_c,E_c)$ in $G$ (line 2). For each node in $V_c$, it takes $O(|\mathcal{T}|)$ time to invoke $\kw{ComputeMSD}$ (line 7), $O(\log |\mathcal{T}|)$ time to compute $\mathcal{MSD}$ (line 8), so the whole time in lines 6-9 is $O(\alpha |\mathcal{T}|)$. For each temporal edge $(w, v, t)$ in $E_c$, \kw{MDC+} will call procedure \kw{UpdateMSD} for at most once, and the cost for each update can be bounded by $O(l)$ according to Lemma~\ref{lemma:updatemsd}.
	Therefore, the total cost for updating all $\mathcal{MSD}$ is bounded by $O(\beta l)$. Putting it all together, the time complexity of Algorithm~\ref{alg:3MaintainDensity} is $O(\alpha |\mathcal{T}|+\beta l)$.
	
	We need to maintain the graph and store collections of $\mathcal{Q},D$ and $deg$ which consumes $O(m)$. Except that, for each node $u$ in $V_c$, we need to store $\mathcal{MSD}[u],\mathcal{MTS}[u], \mathcal{DS}[u]$, which consumes $O(\alpha |\mathcal{T}|)$ in total.
\end{proof}

\section{Algorithms For Mining \skylines}
In this section, we develop an efficient algorithm to record all \skylines.
The basic idea of our algorithm is as follows. The algorithm first only considers the $l$ dimension, and computes the maximal $\widehat{\delta}$, among all the \mdcorealls. Then, the algorithm considers the $\delta$ dimension with $\delta = \widehat{\delta}$ to compute the currently maximal $l'$ value. Using the above method, we can find one \skyline which has the maximal $(l,\delta)$ value of all the skyline communities. The challenge is how to find the other \skyline iteratively. We can tackle this challenge based on the following results.

\begin{lemma}
	\label{lemma:skyline}
	Let $(l',\widehat{\delta})$-$\mdc$ be a \skyline which have the largest $\widehat{\delta}$ among all the \skylines, if the node is not a $(l,\delta)$-dense node with $l > l',\delta>0$, it can not be contained in another \skyline.
\end{lemma}

\begin{proof}
	Suppose that there exists $(l,\delta)$-dense node $v$ whose maximal $l < l'$ in another \skyline $(l^*,{\delta}^*)$-$\mdc$. As $v$ is a $(l,\delta)$-dense node with maximal $l < l'$, according to Definition~\ref{def:4densecore}, $l^*\le l$ . Since $\widehat{\delta}$ is largest among all the \skyline, there holds $\delta^*<\widehat{\delta}$. Therefore, $l^*\le l <l'; \delta^*<\widehat{\delta}$, $(l^*,{\delta}^*)$-$\mdc$ is not a \skyline, which is a contradiction.
\end{proof}


\begin{lemma}
	\label{lemma:skyline2}
	Let $(l',\delta')$-$\mdc$ be a \skyline. If $l^*>l'$ and $(l^*,\delta^*)$-$\mdc$ is another \skyline, $(l^*,\delta^*)$-$\mdc$ must be contained in an induced temporal subgraph from $k$-CORE of $G$ in which $k=\frac{\delta \times l'}{l^*} $.
\end{lemma}

\begin{proof}
	Let $C= (l',\delta')$-$\mdc$ be a \skyline. According to Definition~\ref{def:4densecore}, each node $v\in C$ is an \densenode in $\mathcal{G}(C)$. For each $v$, there exist $S\subseteq C, T\in \mathcal{T}$, satisfying that $\kw{SD}(v,\mathcal{G}_C) \ge \delta'$ and $|T|\ge l'$. If we enlarger $l'$ to $l^*$, in the worst case, the newly added degrees are all zeros, each $v$ will have a segment density $\kw{SD}(v,\mathcal{G}_C) \ge \frac{\delta \times l'}{l^*}$. The remained proof is similar to that of Property~\ref{pro:03}, thus we omit it for brevity.
\end{proof}

Based on Lemma~\ref{lemma:skyline} and Lemma~\ref{lemma:skyline2}, after computing one \skyline \mdcore, as $l$ is integer, we can initialize $l' = l+1$ to get the next \skyline. Furthermore, we can reduce the considering graph by the following corollary.
\begin{corollary}
	\label{coro:lplus1}
	Let $(l,\delta)$-$\mdc$ and $(l',\delta')$-$\mdc$ be two \skylines. If $l'>l$, then nodes in $(l',\delta')$-$\mdc$ must be contained in a  $k$-CORE of $G$ in which $k=\frac{\delta \times l}{l+1} $.
\end{corollary}


\begin{algorithm}[t]
	\scriptsize
	\caption{$\skyline({\mathcal G})$ }
	\label{alg:4skyline}
	\KwIn{Temporal graph $\mathcal{G} = (\mathcal{V},\mathcal{E})$}
	\KwOut{\skylines in $\mathcal{G}$}
	Let $G=(V, E)$ be the de-temporal graph of $ {\mathcal G}$\;
	
	$l\gets 2; \delta\gets 0; R\gets [\emptyset]$;
	$C \gets V$\;
	
	\While{$l\le |\mathcal{T}|$ }{	
		\For {$u \in C$} {
			$(\mathcal{MTS}[u] , \mathcal{DS}[u]) \gets \kw{ComputeMSD}^*({\mathcal G}, l, u, C)$\;
			$\mathcal{MSD}[u] \gets \max(\mathcal{MTS}[u])$\;
			$deg[u] \gets |N_u(G) \cap C|$;
		}
		
		$(\delta, C)\gets \kw{MaxDelta}({\mathcal G},l, C,\mathcal{DS}, \mathcal{MTS}, \mathcal{MSD}, deg) $\;
		$(l, C)\gets \kw{MaxL}({\mathcal G}, l+1,\delta,C,deg) $\;
		$R\gets R\cup (l,\delta,\mathcal{G}_C)$\;
		
		Let $G_c=(V_c, E_c)$ be the $k$-CORE $(k=\frac{\delta \times l}{l+1})$ of $G$\;
		$C \gets V_c$;
		$l\gets l+1$\;
		
	}
	\Return $R$\;
	
	\vspace*{2mm}
	
	{\bf Procedure}  $\kw{MaxDelta}({\mathcal G},l,V^*,\mathcal{DS}, \mathcal{MSD}, \mathcal{MTS},deg )$\\

\While{True}{
	${\cal Q} \gets [\emptyset]; D \gets [\emptyset]; \overline{\delta}\gets \min(\mathcal{MSD}); {\delta}\gets 2nd\min(\mathcal{MSD})$\;
	\For {$u \in V^*$} {
		{\bf if }{$d[u] < \delta$ or $\mathcal{MSD}[u] < \delta$}{ \bf then} ${\cal Q}.push(u)$;
	}
	
	\While{${\cal Q} \neq \emptyset$}{
		$v \gets {\cal Q}.pop(); \ D \gets D \cup \{ v\}$\;
		\For{$w\in N_v(G)\setminus D$, s.t. $deg[w] \ge {\delta}$ and $\mathcal{MSD}[w]\ge {\delta}$}{
			$deg[w]  \gets deg[w] - 1$\;
			
			{\bf if }{ $deg[w] < \delta$}{ \bf then} \{${\cal Q}.push(w)$; {\bf continue;}\}
			
			\For{$ t$, s.t.$(v,w,t) \in \mathcal{E}$}{
				$\mathcal{DS}[w][t] \gets \mathcal{DS}[w][t]-1$\;
				$\mathcal{MSD}[w] \gets \kw{UpdateMSD}(w, t, l,\mathcal{DS}, \mathcal{MTS})$\;
				
			}
			{\bf if }{$\mathcal{MSD}[w] < \delta$ }{\bf then}  \{$\mathcal{Q}.push(w); deg[w] \gets 0;$\}\\
		}
	}
	\If{ $D\neq V^*$}{
			$V^*\gets V^* \setminus D$; {\bf for }$ u \in D${ \bf do }$\mathcal{MSD}[u]\gets \emptyset $ \;
	}{ \bf else}{
			\Return $(\overline{\delta},V^*)$\;
	}
}

	\vspace*{2mm}

{\bf Procedure}  $\kw{MaxL}({\mathcal G}, l, \delta,V^*, deg)$ \\

	\While{$l\le |\mathcal{T}|$}{
	${\cal Q} \gets [\emptyset]; D \gets [\emptyset];\mathcal{MSD} \gets [\emptyset]; \mathcal{MTS} \gets [\emptyset]$\;
	\For {$u \in V^*$} {
		Lines 6-19 in Algorithm~\ref{alg:3MaintainDensity}.
	}
	\If{$D\neq V^*$}{
		$V^*\gets V^* \setminus D$\;
		{\bf if }{$l =|\mathcal{T}|$}{ \bf then }\Return $(l,V^*)$\;
		$l\gets l+1$\;
	}{ \bf else}{
		\Return $(l,V^*)$\;
	}
	
}
\end{algorithm}

The detail of the \skyline algorithm is shown as follows.
First, Algorithm~\ref{alg:4skyline} initializes $l=2,\delta=0$ to be default, $R$ to store the result and $C$ to be the nodes of the considered dense nodes (line 2). Then, the algorithm considers the $l$ dimension and grows $l$ to find all the \skylines. Next, it computes $\mathcal{MSD}[u]$ and $deg[u]$ in the induced graph from nodes $C$ (lines 4-7). By the given $l$, the \kw{MaxDelta} algorithm finds the maximal $\delta$ and the corresponding core nodes (line 8). Next, given one maximal $\delta$, the \kw{MaxL} algorithm finds the maximal $l$ and the final $C$ (line 9). The induced temporal subgraph of $C$ from $\mathcal{G}$ is a \skyline and $(l,\delta,\mathcal{G}_C)$ is recorded as a result (line 10). Based on Corollary~\ref{coro:lplus1}, in the iteration of $l\gets l+1$, the new \skyline must be contained in a induced temporal subgraph from $k$-CORE of $G$ in which $k=\frac{\delta \times l}{l+1}$, so $C$ is updated as $V_c$ for next loop(lines 10-11). The iterations will terminate when $l$ is increased to $|\mathcal{T}|$ (line 3).

Precedence \kw{MaxDelta} describes the process of finding the largest $\delta$ by parameter $l$. It is a loop until all the nodes have been deleted (line 15). The algorithm maintains $\mathcal{Q}$ to be the deleting queue and $D$ to be the deleted nodes. Specifically, it calculates the minimal $\overline{\delta}$ and the second minimal $\delta$ of the $\mathcal{MSD}[u]$ among all nodes (line 16). Then, the nodes are deleted if $deg[w]<\delta$ or $\mathcal{MSD}[w]<\delta$ (lines 19-27). This process are much similar to that in Algorithm~\ref{alg:3MaintainDensity}. Next, if the deleted nodes set $D$ is not equal to the remained nodes set $V^*$, the remained $V^*$ is updated by $V^* \setminus D$ and $\mathcal{MSD}$ will pop all the $\mathcal{MSD}[u]$ for $u$ in the deleted nodes' set $D$ (lines 28-29). Else, if $D = V^*$, then the remained nodes $V^*$ will have maximal $\overline{\delta}$ (lines 30). Furthermore, precedence \kw{MaxL} can use the remained nodes set of  \kw{MaxDelta} and the known maximal $\delta$ to find the maximal $l$. It grows $l$ to find the largest $l$ and it will terminate if $l$ increases to $|\mathcal{T}|$ (line 32). The unsatisfying nodes are deleted same as that in Algorithm~\ref{alg:3MaintainDensity} (lines 34-35). \kw{MaxL} ends at the first time when all the $V^*$ will be deleted or $l=|\mathcal{T}|$, and it returns $l$ at this time and the remained nodes set $V^*$ (lines 36-40).

\stitle{Complexity of Algorithm~\ref{alg:4skyline}. } The worst time and space complexity of Algorithm~\ref{alg:4skyline} are $O(m |\mathcal{T}|^2)$ and $O(n |\mathcal{T}|+m)$ respectively.
However, the pruning rule based on Corollary~\ref{coro:lplus1} can reduces the computation time greatly. We will show the running time in practice at Section V.

%
%
%
%
%

\section{Experiments} \label{sec:exp}
In this section, we conduct extensive experiments to evaluate the effectiveness and efficiency of the proposed algorithms. We implement seven different algorithms for comparison: \kcore, \maxdensesub~\cite{17tkddDynamicDense}, $\mdcb$, \mdc, $\mdcplus$, \skyline, \skylineb. \kcore is a baseline which computes the $k$-CORE $(k=\delta)$ of the de-temporal graph $G$. $\mdcb$ is another baseline which computes \mdcore using the framework shown in Algorithm~\ref{alg:1densecore}, but it enumerates all subsequences to compute \lmsdall. \maxdensesub \cite{17tkddDynamicDense} is also a baseline algorithm which can find the densest subgraph in a temporal graph. \mdc is the implementation of Algorithm~\ref{alg:1densecore} that uses Algorithm~\ref{alg:2ComputeDensity} to compute  $\mathcal{MSD}$. $\mdcplus$ is the implementation of Algorithm~\ref{alg:3MaintainDensity} to compute \mdcore. \skyline can output all the \skylines by Definition~\ref{def:5skyline} and it is an implementation of Algorithm~\ref{alg:4skyline}. \skylineb is a basic implementation of \skyline without integrating the pruning rules developed in Corollary~\ref{coro:lplus1}.

All algorithms are implemented in Python and the source code is available at \url{https://github.com/VeryLargeGraph/MDC}. All the experiments are conducted on a server of Linux kernel 4.4 with Intel Core(TM) i5-8400@3.80GHz and 32 GB memory. 

\begin{table}[t!]\vspace*{-0.5mm}
	\scriptsize
	\centering
	\caption{Statistics of datasets} \label{table:datasets}
	\vspace*{-0.5mm}
	\setlength{\tabcolsep}{0.8mm}{
		\begin{tabular}{c|c|c|c|c|c|c}
			\hline
			Dataset & $|V|=n$ & $| E|=m^\prime$ & $| \mathcal{E}|=m$& $d_{\max}$&${|\mathcal T|}$& ${\scriptsize \kw{Time \ scale}}$ \\ \hline
			\chess	&7,301		&55,899		&63,689		&233 	&101	&month\\
			\lkml  	&26,885		&159,996	&328,092	&14,172	&96	&month\\
			\enron  &86,978		&297,456	&499,983	&2,164	&48
				&month\\
			\dblp  	&1,729,816	&8,546,306	&12,007,380	&5,980	&78	&year\\
			\youtube&3,223,589	&9,376,594	&12,218,755	&129,819	&225	&day\\
			\flickr &2,302,925	&22,838,276	&24,690,648	&28,276	&197	&day\\
			\hline
			\hline
			\mathoverflow &24,759	&187,986 &294,293&5,556	&2,351	&day\\
			\askubuntu  	&157,222	&455,691 &549,914&7,325	&2,614 	&day\\
			\wikitalk  		&1,094,018& 2,787,967	&4,010,611&214,518	&2,321	&day\\ 	
			\hline
		\end{tabular}
	}
\end{table}\vspace*{-0.5mm}

\stitle{Datasets.} We use 9 different real-world temporal networks in the experiments. The detailed statistics of our datasets are summarized in Table~\ref{table:datasets}, where $d_{\max}$ denotes the maximum number of temporal edges associated with a node, and ${|\mathcal T|}$ denotes the number of snapshots. 
All the snapshots are simple, undirected and unweighted graphs.  \chess\footnote[1]{http://konect.uni-koblenz.de/networks/} is a network that represents two chess players playing game together from 1998 to 2006.
$\lkml^1$ is a communication network of the Linux kernel mailing list from 2001 to 2011.
$\enron^1$ is an email communication network between employees of Enron from 1999 to 2003.
\dblp \footnote[2]{https://dblp.uni-trier.de/xml/} is a collaboration network of authors in \dblp from 1940 to Feb. 2018.
$\kw{Youtube}^3$ ($\youtube$ for short) and $\kw{Flickr}^1$ ($\flickr$) are friendship networks of users in \kw{Youtube} and \kw{Flickr}, respectively.
$\kw{MathOverflow}\footnote[3]{http://snap.stanford.edu/data/index.html}$ (\mathoverflow), $\kw{AskUbuntu}^3$ ($\askubuntu$) are temporal networks of interactions on the stack exchange web site \url{mathoverflow.net} and \url{askubuntu.com}, respectively.
$\kw{WikiTalk}^3$ ($\wikitalk$) is a temporal network representing the interactions among Wikipedia users.


\stitle{Parameter settings.} There are two parameters $l$, $\delta$ in the \mdcore model. For the parameter $l$, we vary it from 3 to 11 with a default value of 3 in the testing. We vary $\delta$ from 3.0 to 11.0 with a default value of 3.0. Unless otherwise specified, the values of the other parameters are set to their default values when varying a parameter.

\stitle{Goodness Metrics.} Since most existing metrics (e.g., modularity) for measuring the community quality are tailored for traditional graphs, we introduce two goodness metrics evaluating communities for temporal graphs, which are motivated by  \emph{density} and \emph{separability} \cite{12icdmgroundtruthcommunity}. Let $C$ be a community computed by different algorithms.

\textit{Average Density (\ad)} builds on intuition that good communities are well connected. It measures the fraction of the temporal edges that appear between the nodes in $C$:
$\ad\triangleq [\frac{ \sum_{v_i\in C} deg_{\mathcal{G}_C}(v_i)}{|{C}|} ] $, where $deg_{\mathcal{G}_C}(v_i)$ denotes the number of temporal edges that are associated with $v_i$ in the community $C$.

\textit{{Average Separability (\as)}} captures the intuition that good communities are well-separated from the rest of the network, meaning that they have relatively few across edges between $C$ and the rest of the network:
$\as \triangleq [\frac{ |\{(u, v, t) \in \mathcal{E} : u \in C,  v \in C\}| / |C|}     {|S=\{(u, v,t ) \in \mathcal{E} : u \in C,  v \notin C\}| /|S| } ] $, which measures the ratio between the internal average density and external average density.

\begin{figure}[t]
	\centering
	\subfigure[\ad]{
		\includegraphics[height=3.0cm]{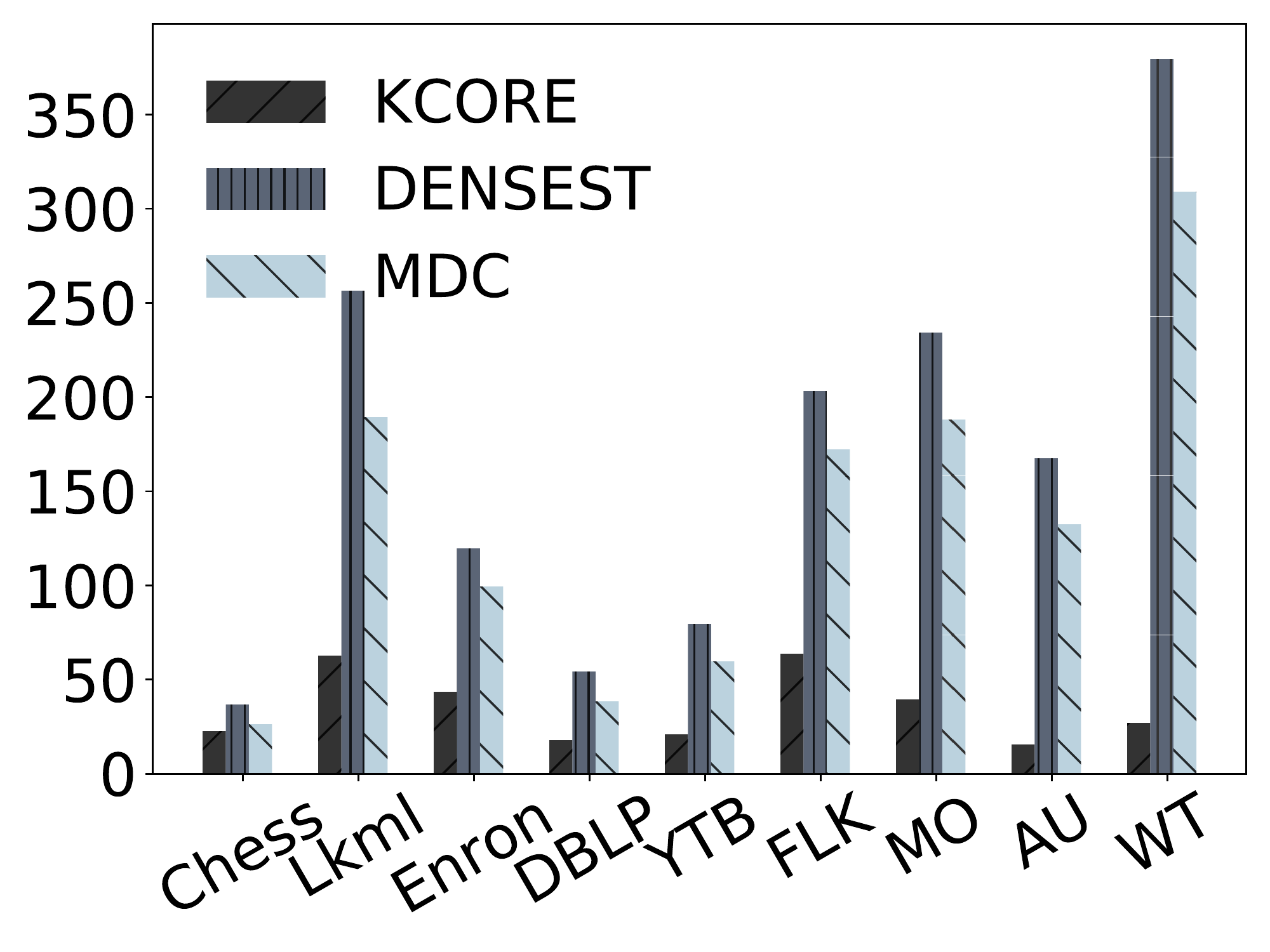}
	}
	\subfigure[\as]{
		\includegraphics[height=3.0cm]{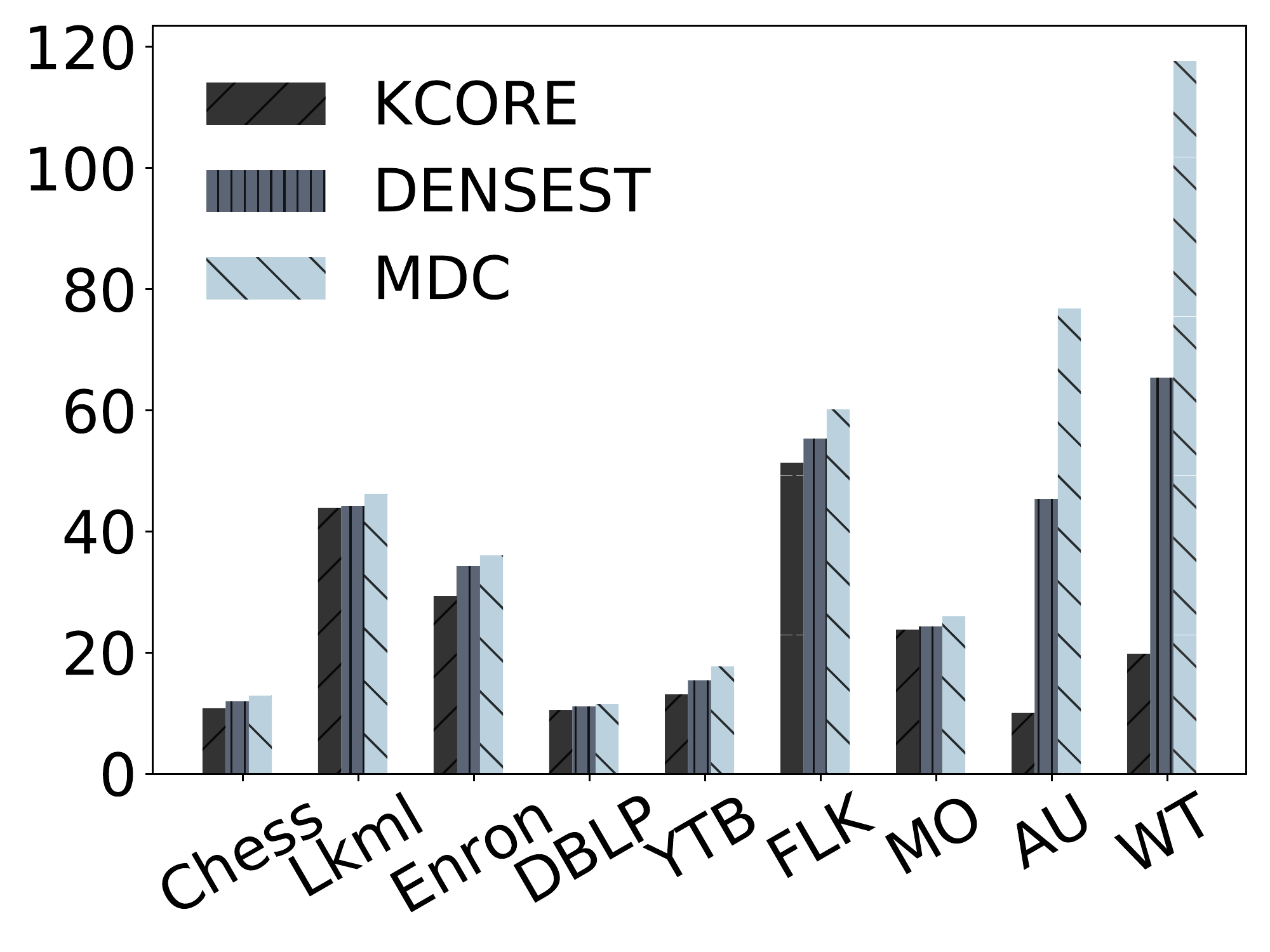}
	}
	
	\vspace*{-0.3cm}\caption{Effectiveness results of \kcore, \maxdensesub and \mdc}
	\vspace*{-0.3cm}
	\label{exp-1}
\end{figure}

\begin{figure}[t]
	\centering
	\subfigure[vary $l$ ($\dblp$)]{
		\includegraphics[height=2.8cm]{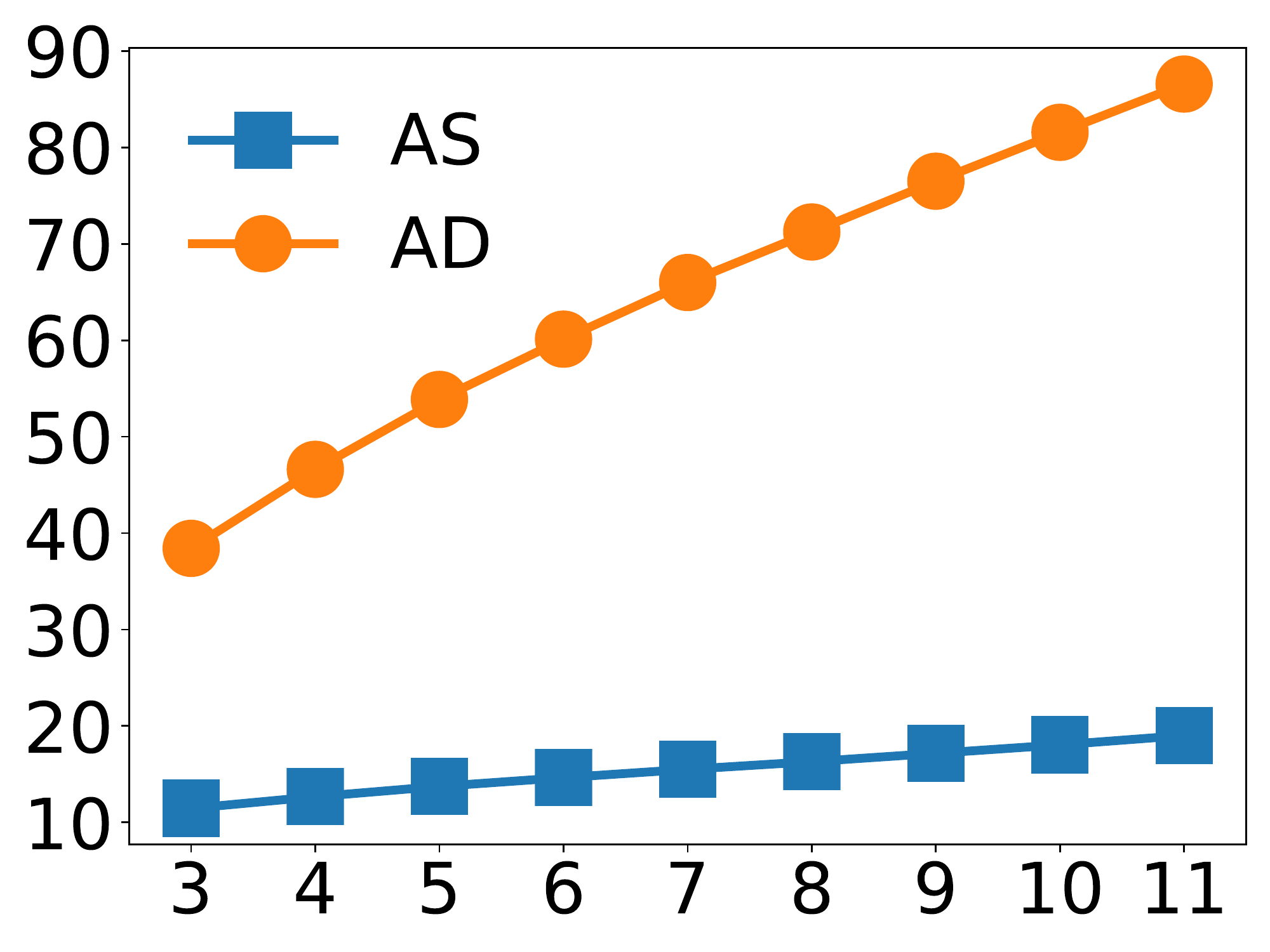}
	}
	\subfigure[vary $\delta$ ($\dblp$)]{
		\includegraphics[height=2.8cm]{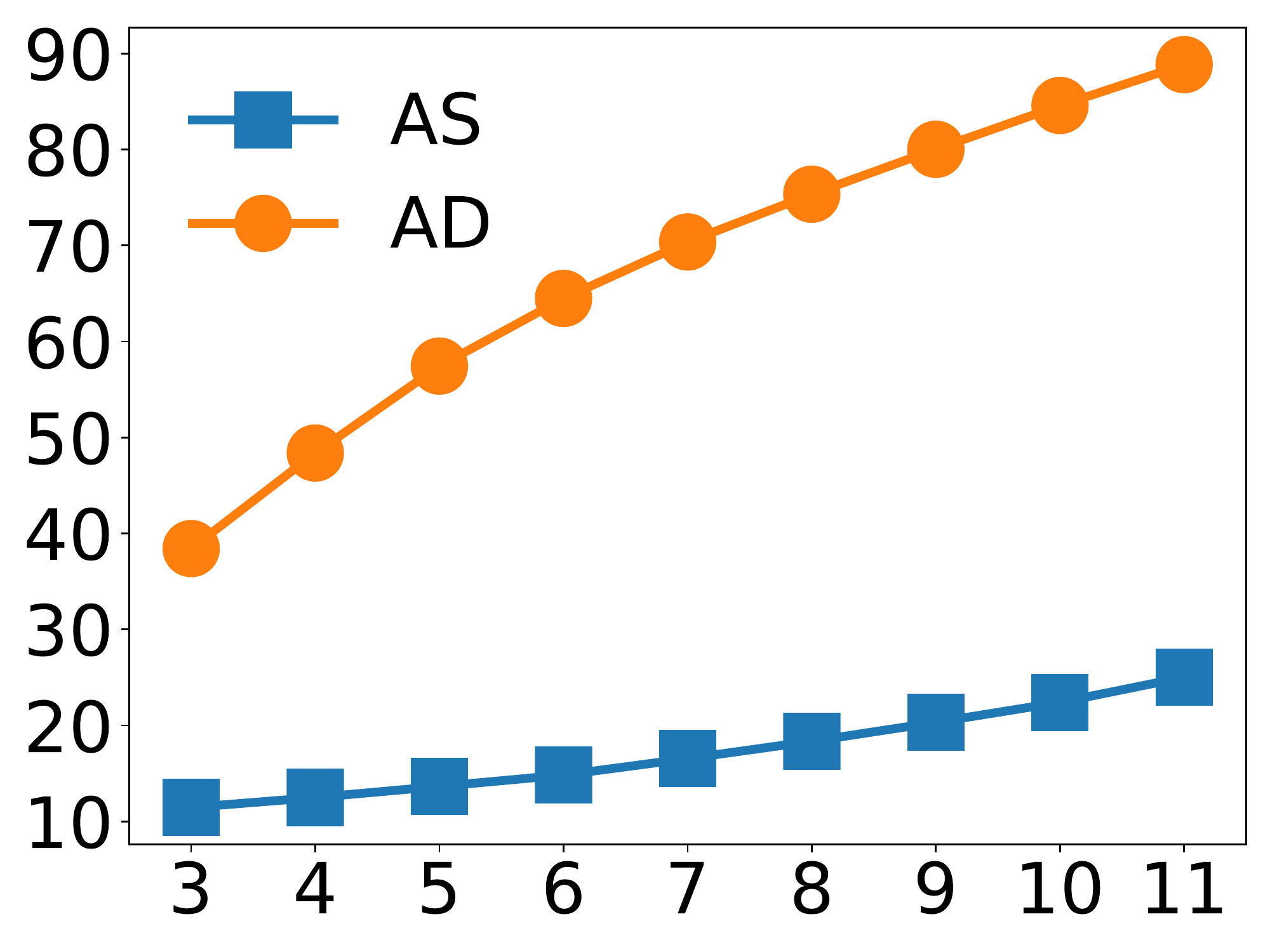}
	}
	\vspace*{-0.2cm}\caption{Effectiveness of \mdc with varying parameters on \dblp}
	\vspace*{-0.5cm}
	\label{exp-2}
\end{figure}

\subsection{Effectiveness Testing} \label{subsec:effectiveness}

\stitle{Exp-1. Effectiveness of \kcore, \maxdensesub and \mdc.} Fig.~\ref{exp-1} shows the qualities of the communities computed by different algorithms under the default parameter setting. Similar results can also be observed using the other parameter settings. As can be seen in Fig.~\ref{exp-1}(a), \maxdensesub significantly outperforms the others in terms of the \ad metric. We also observe that \maxdensesub obtains the subgraph with the largest density. Both {\maxdensesub} and {\mdc} perform much better than {\kcore}. We can see that the \ad values for both {\maxdensesub} and {\mdc}  in \wikitalk is much larger than those in the other datasets. The reason is that the maximum degree in \wikitalk is the largest one among all datasets, thus there must exist a community with higher density. In Fig.~\ref{exp-1}(b), the \mdc community proposed by us have higher \as value among all datasets. Compared to the other datasets, the \ad value on \mathoverflow is high but the \as value is low. The reason is that \as metric captures the ratio between the internal average density and external average density. Clearly, each node in \mdc has a high internal average density.

\stitle{Exp-2. Effectiveness results with varying parameters.} Here we study how the parameters affect the effectiveness performance of our algorithm.  Fig.~\ref{exp-2} shows the results of {\mdc} with varying parameters on \dblp. Similar results can also be observed on the other datasets. As can be seen, both \as and \ad values increase with growing $l$ and $\delta$. The reason is that the lasting time of the \mdc increases when $l$ increases, and the average density of nodes in \mdc increases when $\delta$ increases. 

\begin{figure}[t]
	\centering
	\subfigure[\dblp]{
		\includegraphics[height=3cm]{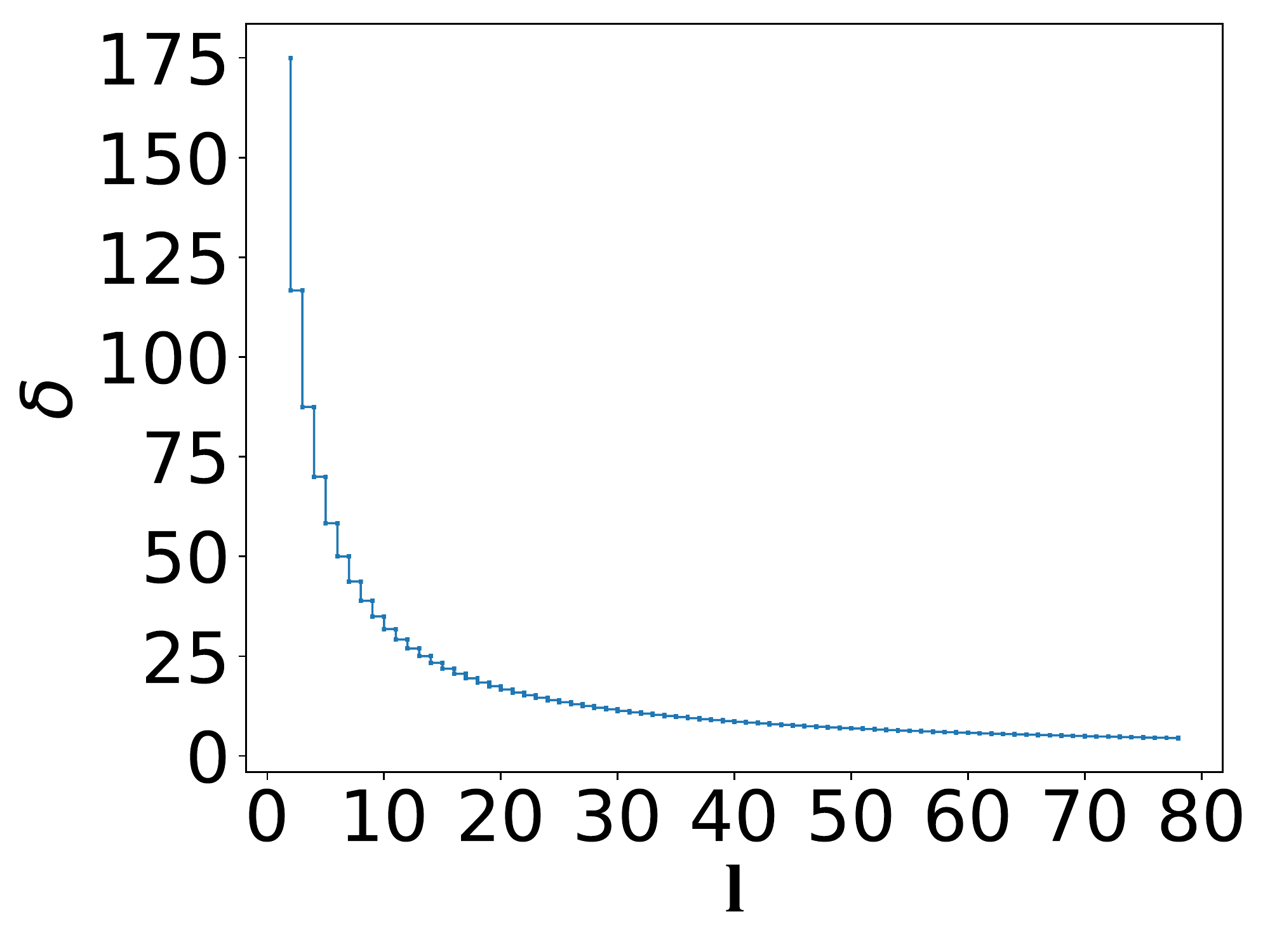}
	}
	\subfigure[\lkml]{
		\includegraphics[height=3cm]{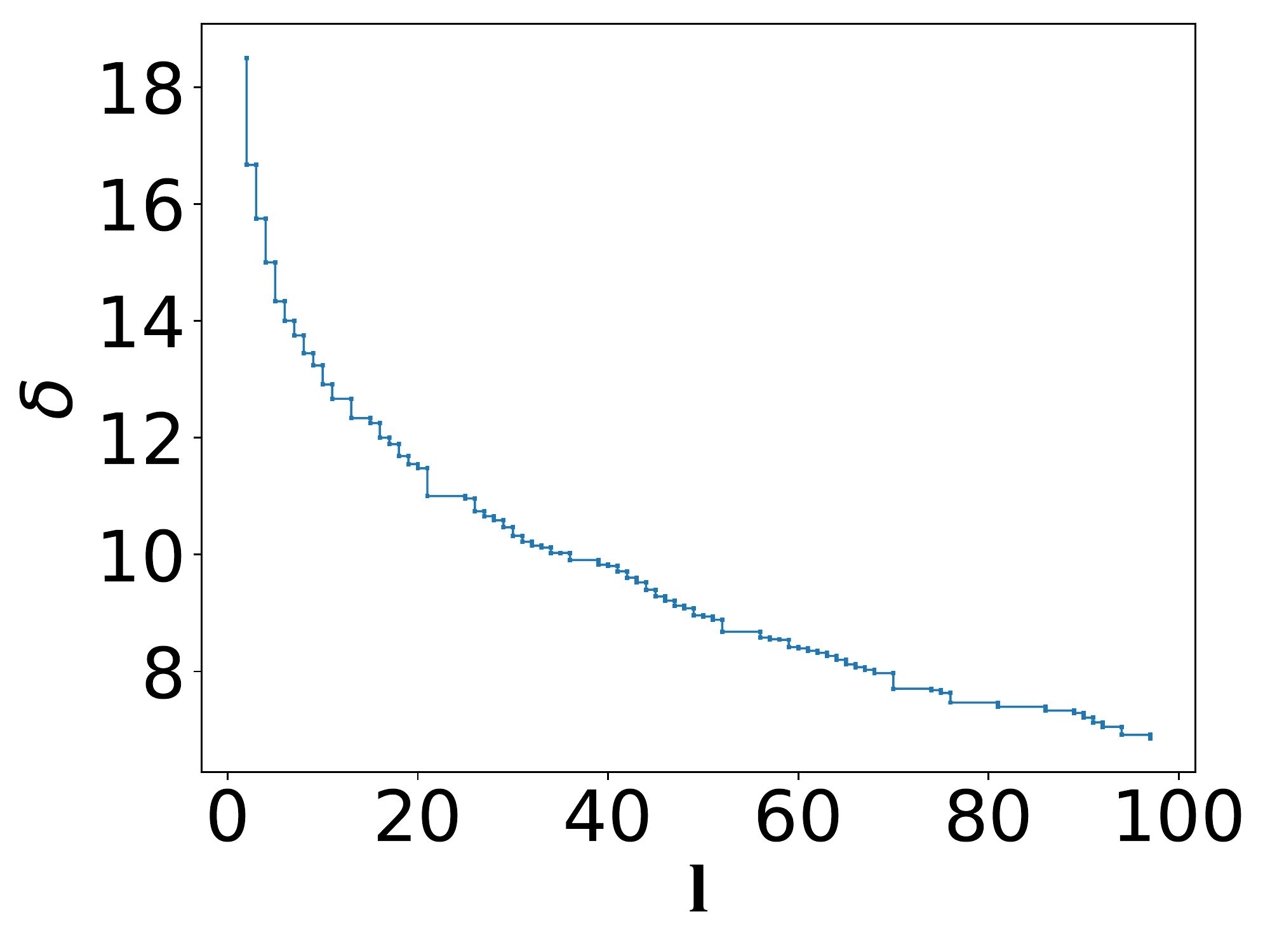}
	}
	\vspace*{-0.3cm}\caption{$l,\delta$ values of each \skyline on different datasets}
	\vspace*{-0.3cm}
	\label{exp-3}
\end{figure}

\begin{figure}[t]
	\centering
	\subfigure[\as]{
		\includegraphics[height=3cm]{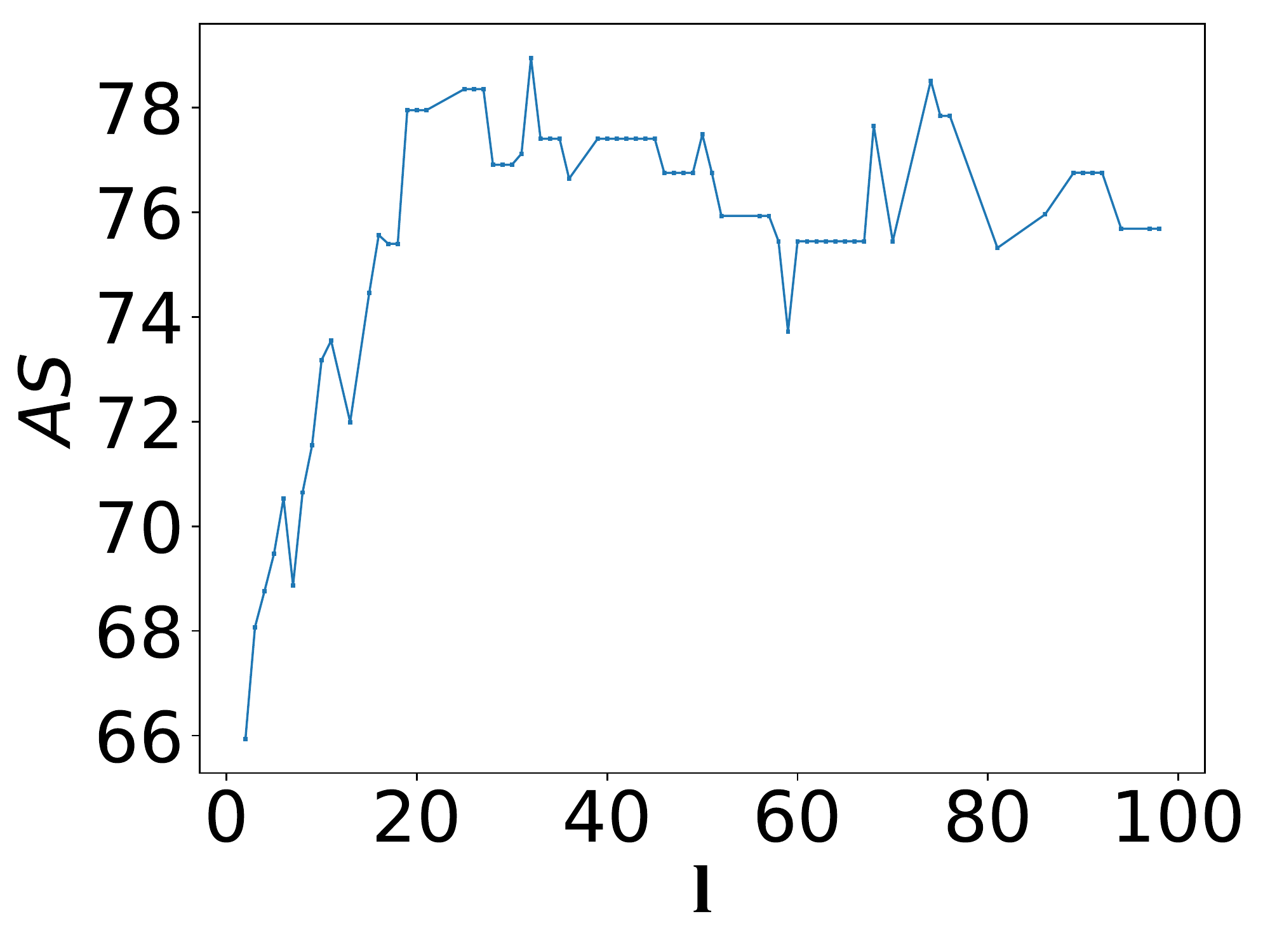}
	}
	\subfigure[\ad]{
		\includegraphics[height=3cm]{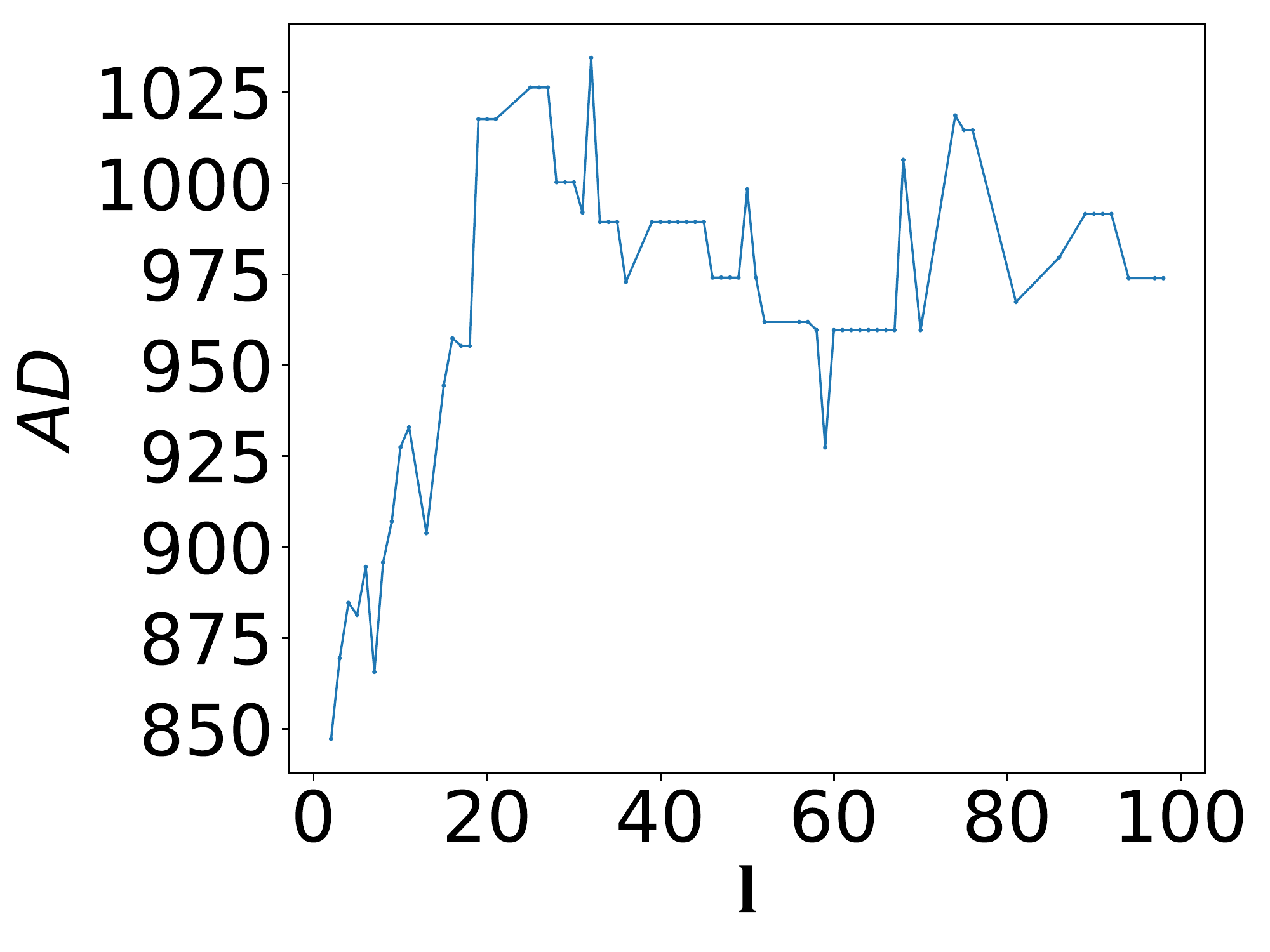}
	}
	\vspace*{-0.3cm}\caption{Effectiveness results of \skylines in \lkml}
	\vspace*{-0.5cm}
	\label{exp-4}
\end{figure}

\stitle{Exp-3. Results of \skyline.} Fig.~\ref{exp-3} shows the $l,\delta$ values for each \skyline on \dblp and \lkml. Again, similar results can also be observed on the other datasets. From Fig.~\ref{exp-3}(a), we observe that when $l=2$, an \mdcore in \dblp achieves the \lmsdall which is equal to 175. The $\delta$ values of the \skylines  drop dramatically when $l=10$. This is because most researchers in \dblp typically cooperate with each others in a continuous time of 2-10 years. As desired, both  Fig.~\ref{exp-3}(a) and Fig.~\ref{exp-3}(b) exhibit a staircase shape because of the parato-optimal property. Fig.~\ref{exp-4} shows the \as, \ad values of \skylines on \lkml. The results on the other datasets are consistent. We can see that the \as and \ad values increase as $l$ increases from 0 to 20, and then \as and \ad change slightly as $l$ increases from 20 to 100. This is because real-world bursting communities can only last in a short time.



\stitle{Exp-4. Case study on \enron.} The \enron dataset consists of the emails sent between employees of Enron from 1999 to 2003. Enron was an energy-trading and utilities company based in Houston, Texas, that perpetrated one of the biggest accounting frauds in history. Enron's executives employed accounting practices that falsely inflated the company's revenues and, for a time, made it the seventh-largest corporation in the United States. Once the fraud came to light, the company quickly unraveled, and it filed for bankruptcy on Dec. 2, 2001. Fig~\ref{exp-case}(a) shows the part of \kcore in a subgraph in which each employee sends e-mails in year 2001. The model of \kcore performs very bad, as the resulting community involves large numbers of employees, so it is hard to find the employees who are significant in the company. Fig~\ref{exp-case}(b) shows a part of \mdc with parameters $(l=3, \delta=3) $. We can see that the employees in this subgraph are annotated by the $l$-segment with maximum density which is a continuous time of at least 3 months. In addition, we can find that the actual timestamps in the $l$-segment of nodes in \mdc are around Dec, 2001. Therefore, the employees in \mdc must be the key persons in Enron, and they are responsible for the bankruptcy of Enron.

\begin{figure}[t]\vspace*{-0.3cm}
	\centering
	\subfigure[\kcore]{
		\includegraphics[height=3.3cm]{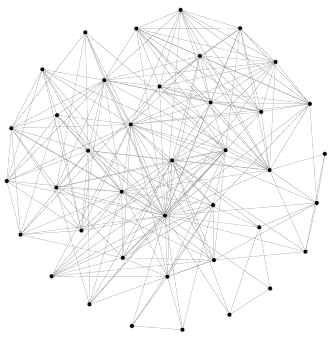}
	}
	\subfigure[\mdc]{
		\includegraphics[height=3.3cm]{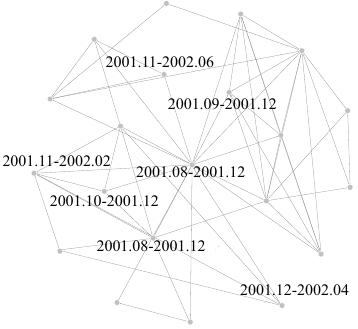}
	}
	\vspace*{-0.3cm}\caption{Case study on \enron}
	\vspace*{-0.3cm}
	\label{exp-case}
\end{figure}

\subsection{Efficiency Testing} \label{subsec:exp-efficiency}


\stitle{Exp-5. Running time of the algorithms.} Table.~\ref{table:runtime} evaluates the running time of \kcore, \mdcb, \mdc, $\mdcplus$ with parameters $l=3, \delta =3$. Similar results can also be observed with the other parameter settings. From Table.~\ref{table:runtime}, we can see that $\mdcplus$ is much faster than \mdcb and \mdc on all datasets. Note that \kcore is the fastest algorithm, as it has linear time complexity of~\cite{15bigdatatemporalcore}. But \kcore is ineffective to find bursting communities. For example, on \dblp, \kcore takes 6.84 seconds and our proposed $\mdcplus$ only consumes 26.95 seconds. On \wikitalk, we can see that \mdcb takes 11865.87 seconds to compute the \mdcore and $\mdcplus$ only takes 57.65 seconds. These results confirm that our proposed algorithms are indeed very efficient on large real-life temporal networks.

\begin{table}[t!]\vspace*{-0.3mm}
	\scriptsize
	\centering
	\caption{Running time (s) of different algorithms with $l=3,\delta=3$} \label{table:runtime}
	\vspace*{-0.5mm}
		\begin{tabular}{c|c|c|c|c}
			\hline
			Dataset & \kcore & $\mdcb$ & \mdc &$\mdcplus$  \\ \hline
			\chess	&0.05 &1.32 &0.78& 0.50 \\
			\lkml  	&0.06 & 2.4 &  1.02& 0.36 \\
			\enron  &0.19&13.41& 3.54& 1.25 \\
			\dblp  	&6.84&187.32 &53.90&26.95 \\
			\youtube &30.53 & 759.52 &126.92& 68.23\\
			\flickr &17.53&876.4&122.87&34.52 \\
			\mathoverflow &0.11 &1200.23&30.15&3.71\\
			\askubuntu  	&0.52&2599.78&66.89&13.36\\
			\wikitalk  		&2.15&11865.87&145.23&57.65\\ 	
			\hline
		\end{tabular}
\end{table}\vspace*{-0.5mm}

\stitle{Exp-6. Running time of computing all \skylines.}
Fig.~\ref{exp-6time} shows the running time of \skylineb and \skyline with the default parameter setting. We can see that
\skyline is much faster than \skylineb on all datasets. For example, \skyline needs around 5,320 seconds and 18,680 seconds to compute all the \skylines in \mathoverflow and \askubuntu datasets which cuts the running time over \skylineb by 139\% and 131\%, respectively. Note that both \skylineb and \skyline cannot obtain results on \wikitalk in 1 day. These results indicate that the pruning rule in Corollary~\ref{coro:lplus1} is indeed very powerful in practice.


\begin{figure}[t]\vspace*{-0.5mm}
	\centering
	\includegraphics[height=3.7cm]{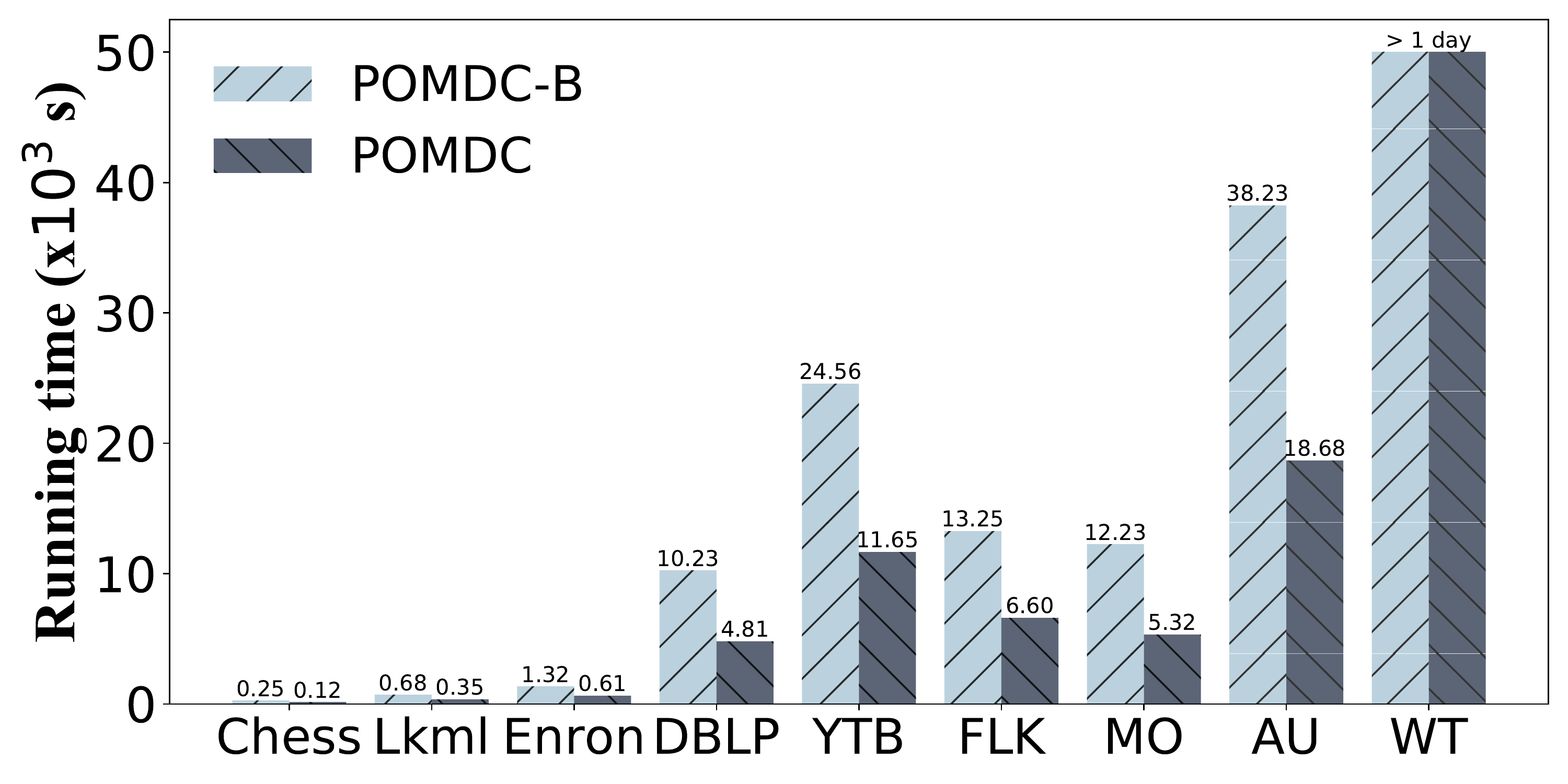}

	\vspace*{-0.3cm}\caption{Running time of \skyline V.S. \skylineb}
	\vspace*{-0.3cm}
	\label{exp-6time}
\end{figure}

\begin{figure}[t]
	\centering
	\subfigure[vary $l$ ($\dblp$)]{
		\includegraphics[height=2.8cm]{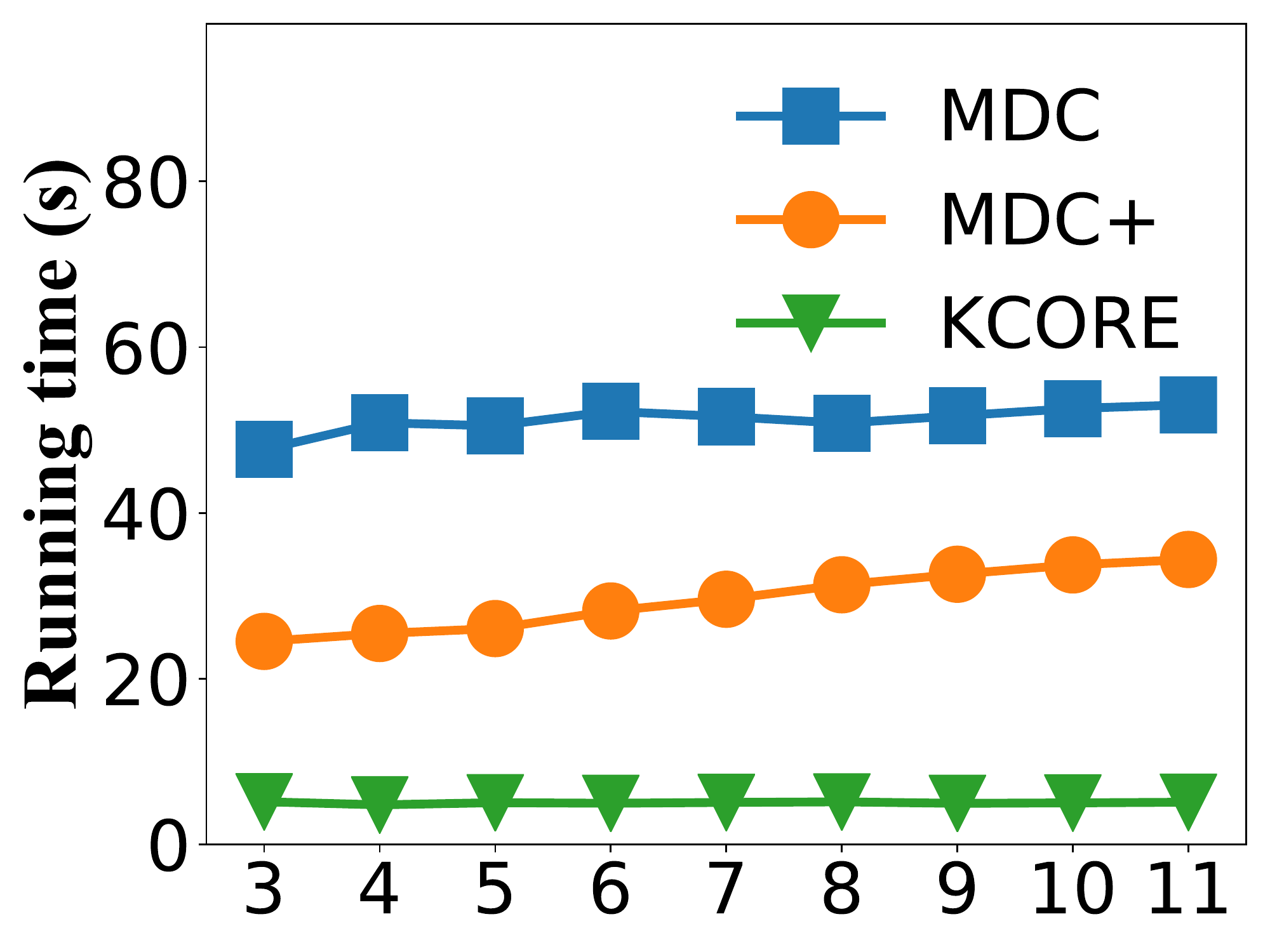}
	}
	\subfigure[vary $\delta$ ($\dblp$)]{
		\includegraphics[height=2.8cm]{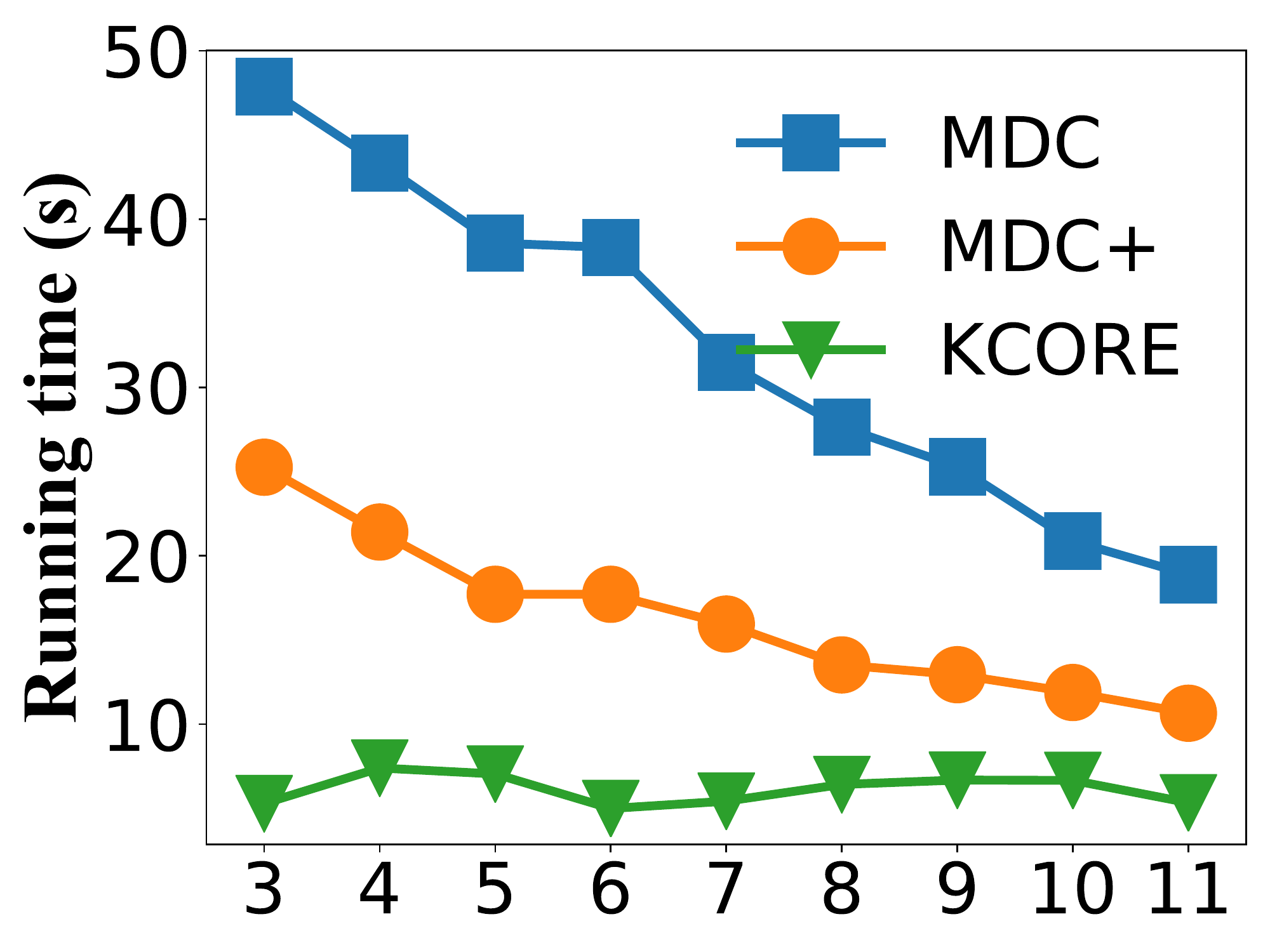}
	}
	\vspace*{-0.3cm}\caption{Running time of different algorithms on \dblp with varying $l,\delta$ }
	\vspace*{-0.3cm}
	\label{exp-5time}
\end{figure}

\begin{figure}[t]
	\centering
	\subfigure[percents of edges]{
		\includegraphics[height=2.8cm]{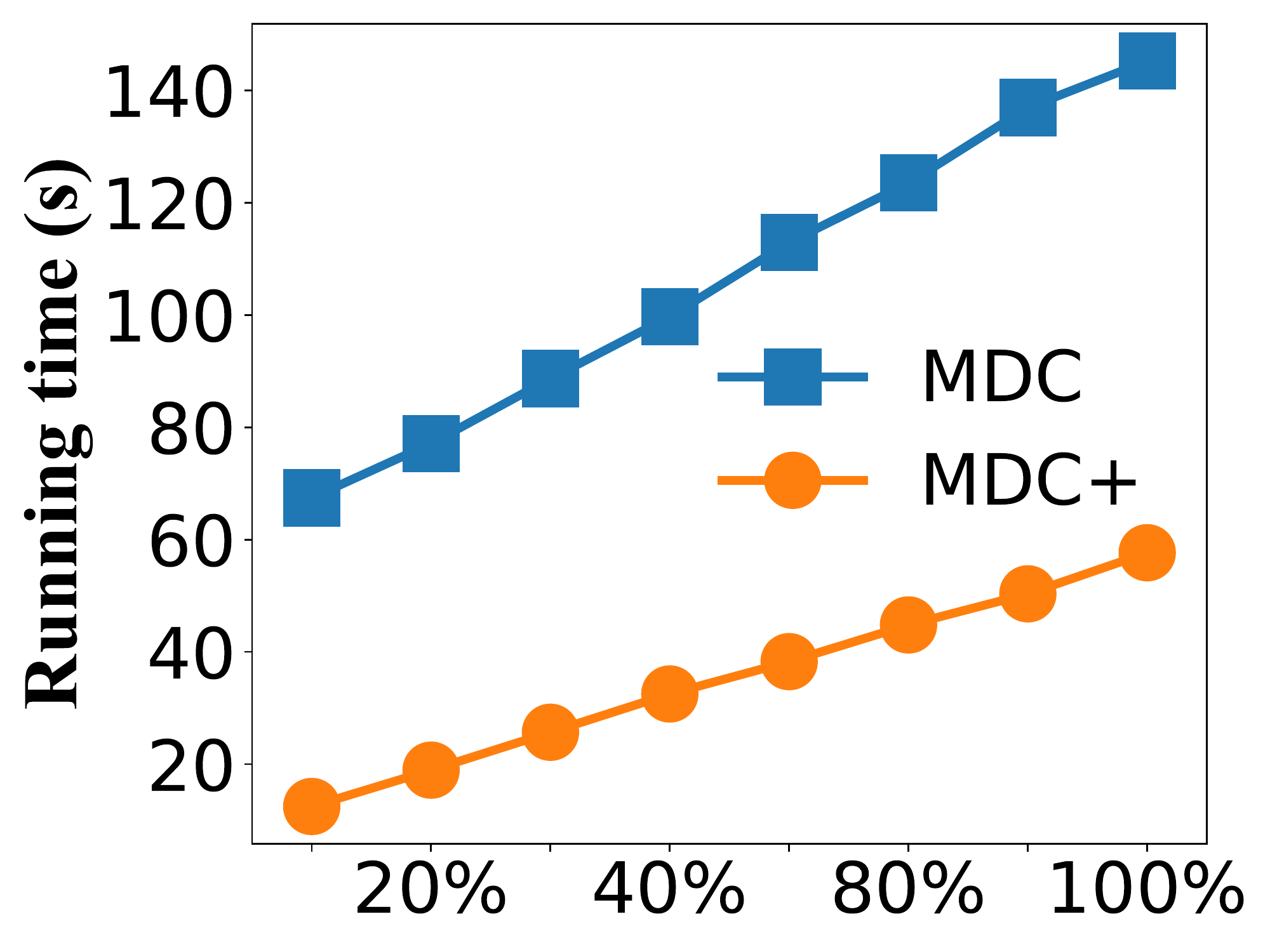}
	}
	\subfigure[percents of $|\mathcal{T}|$]{
		\includegraphics[height=2.8cm]{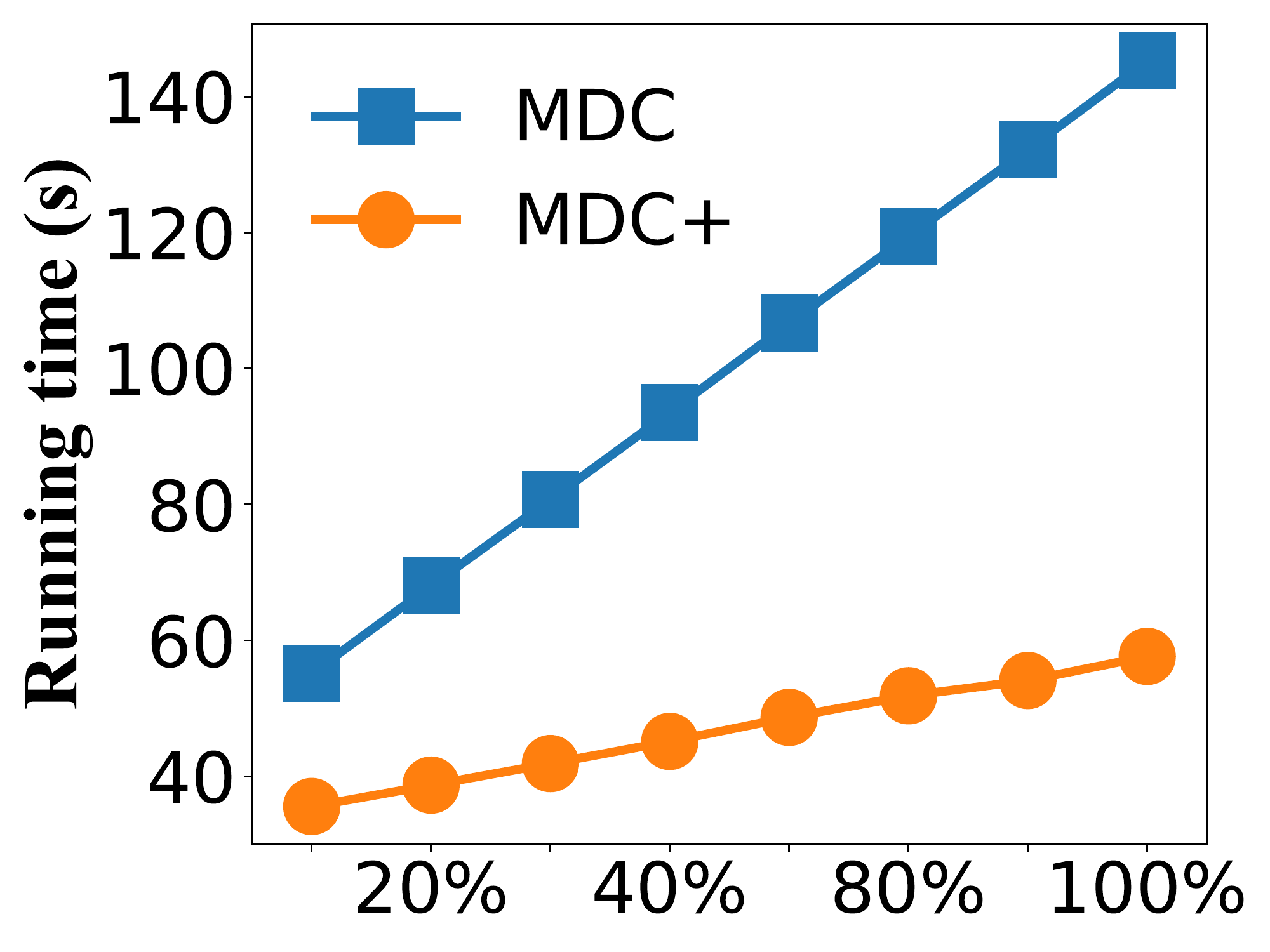}
	}
	\vspace*{-0.3cm}\caption{Scalability testings on \wikitalk }
	\vspace*{-0.3cm}
	\label{exp-8time}
\end{figure}

\stitle{Exp-7. Running time with varying parameters.} Fig.~\ref{exp-5time} shows the running time of \kcore, \mdc and $\mdcplus$ with varying parameters on \dblp. Similar results can also be observed on the other datasets. As can be seen, $\mdcplus$ is faster than \mdc under all parameter settings. In Fig.~\ref{exp-5time}(a), the running times of \kcore and \mdc remain unchanged, but the running time of $\mdcplus$ increases slowly with an increasing $l$. These results confirm that the time complexity of \kcore and \mdc is independent to $l$, and the time complexity of $\mdcplus$ is linear w.r.t.\ $l$.
We also see that the running time of $\mdcplus$ and \mdc decrease with an increasing $\delta$, because all of them need to reduce the graph by the $k$-CORE based on Property 3.3 and the size of $k$-CORE decreases as $\delta$ increases. 

%
%

\begin{table}[t!]
	\scriptsize
	\centering
	\caption{Memory overhead of \mdc and $\mdcplus$} \label{table:memorycost}
	\vspace*{-0.1cm}
	\setlength{\tabcolsep}{2 mm}{
		\begin{tabular}{c|ccc}
			\hline
			& {\scriptsize Graph in Memory} & {\scriptsize Memory of \mdc} & {\scriptsize Memory of $\mdcplus$}\\ \hline
			\chess		&3.5MB		&9.2MB		&44.2MB		\\
			\lkml  			&20.1MB		&44.4MB		&121.2MB		\\
			\enron  		&53.3MB		&107.6Mb		&303.2MB		\\
			\dblp  			&1,089.5MB	&2,328.2MB	&3,934.3MB	\\
			\youtube  			&698.5MB	&1,452.8MB	&3,318.1MB	\\
			\flickr  			&1,375.5MB	&3,198.2MB	&5,647.2MB	\\
			\mathoverflow  			&13.23MB	&45.23MB	&92.75MB	\\
			\askubuntu  			&50.23MB	&140.32MB	&459.2MB	\\
			\wikitalk  			&324.5MB	&1023.23MB	&3,163.2MB	\\
			\hline
		\end{tabular}
	}
	\vspace*{-0.3cm}
\end{table}

\stitle{Exp-8. Scalability.} Fig.~\ref{exp-8time} shows the scalability of \mdc and $\mdcplus$ on \wikitalk dataset. Similar results can also be observed on the other datasets. We generate ten temporal subgraphs by randomly picking 10\%-100\% of the temporal edges or 10\%-100\% of the timestamps, and evaluate the running times of \mdc and $\mdcplus$ on those subgraphs.
As shown in Fig.~\ref{exp-8time}, the running time increases smoothly with increasing number of edges or increasing size of $|\mathcal{T}|$. These results suggest that our proposed algorithms are scalable when handling large temporal networks.

\stitle{Exp-9. Memory overhead.} Table~\ref{table:memorycost} shows the memory usage of \mdc and $\mdcplus$ on different datasets. We can see that the memory usage of \mdc and $\mdcplus$ is higher than the size of the temporal graph, because \mdc only needs to store $deg[u]$ (for each node $u$) but $\mdcplus$ needs to store $\mathcal{MSD}[u],\mathcal{MTS}[u], \mathcal{DS}[u]$ (for each node $u$). In practice, we can free memory of $\mathcal{MSD}[u],\mathcal{MTS}[u], \mathcal{DS}[u]$ once $u$ has been added into the deleting queue $\mathcal{Q}$. Therefore, on large datasets, the memory usage of $\mdcplus$ is typically lower than ten times of the size of the temporal graph. For instance, on \wikitalk, $\mdcplus$ consumes 3,163.2MB memory while the graph needs 324.5MB. These results indicate that \mdc and $\mdcplus$ achieve near linear space complexity, which confirms our theoretical analysis in Sections III.

\section{Related Work}
\stitle{Dense subgraph mining in temporal graphs.}
Our work is related to the problem of mining dense subgraphs in temporal graph. Ma et al.\ \cite{17icdedensegraphtemporal} and Bogdanov et al.\ \cite{icdm11heavy} investigated the dense subgraph problem in weighted temporal graphs. Rozenshtein et al.\ \cite{18icdmSegmentation} studied the problem of mining dense subgraphs at different time  intervals, and they also considered a problem of finding the densest subgraph in a temporal network \cite{17tkddDynamicDense}.
Liu et al.\ \cite{icde19stochastic} proposed a novel stochastic approach to find the densest lasting subgraph.
Many other works \cite{www05dense,stoc15dense} aim at maintaining the average-degree densest-subgraph in a graph streaming scenario. Unlike all these studies, we focus mainly on detecting bursting communities in temporal graphs. Note that the above mentioned dense subgraphs are not bursting communities because not all nodes in a dense subgraph are bursting in a period of time.

\stitle{Temporal graph analysis.} The problem of temporal graph analysis has attracted much attention in recent years.
Yang et al.\ \cite{14wwwjevolvcomm} proposed an algorithm to detect frequent changing components in temporal graph.
Huang et al.\ \cite{15sigmodmsttemporalgraph} investigated the minimum spanning tree problem in temporal graphs. Gurukar et al.\ \cite{15sigmodmotiftemporal} presented a model to identify the recurring subgraphs that have similar sequence of information flow. Wu et al.\ \cite{16icdetemporalreach} proposed an efficient algorithm to answer the reachability query on temporal graphs.
Yang et al.\ \cite{16kddtemporalquasiclique} studied a problem of finding a set of diversified quasi-cliques from a temporal graph.
Wu et al.\ \cite{15bigdatatemporalcore} and Galimberti et al.\ \cite{cikm18spancores} studied the core decomposition problem in temporal networks. Li et al.\ \cite{18persistent} developed an algorithm to detect persistent communities in a temporal graph. More recently, Qin et al.\ \cite{19ICDEperiodicclique} proposed a periodic clique model to mine periodic communities in a temporal graph. To the best of our knowledge, we are the first to study the problem of mining bursting communities in temporal graph.

\stitle{Community mining on traditional and dynamic graphs.}
Community mining is a problem of identifying cohesive subgraphs from a graph. Notable cohesive subgraph models include maximal clique \cite{Cheng2011Maximal}, quasi clique\cite{tsourakakis2013denser}, $k$-core\cite{12tkdecoremaintain,sigmod18Distancecore} and $k$-truss \cite{11icdekcore, 14sigmodtrusscommunity}.
There are a number of studies for mining communities on dynamic networks \cite{Community2018surveydynamic}.
Lin et al.\ \cite{lin2008facetnet} proposed a probabilistic generative model for analyzing communities and their evolutions.
Chen et al.\ \cite{ICDMW2010Tracking} tracked community dynamics by introducing graph representatives.
Agarwal et al.\ \cite{pvldb2012DenseinDynamicGraphs} studied how to find dense clusters efficiently for dynamic graphs in spite of rapid changes to the microblog streams. Li et al.\ \cite{12tkdecoremaintain} devised an algorithm which can maintain the $k$-core in large dynamic graphs. Most community detection studies on dynamic graphs aims to maintain communities that evolve over time. Unlike these studies, we aim to detect bursting communities in temporal graphs.

\section{Conclusion}
In this work, we study a problem of mining bursting communities in a temporal graph. We propose a novel model, called \mdcore, to characterize the bursting communities in a temporal graph. To find all \mdcores, we first develop an dynamic programming algorithm which can compute the segment density efficiently. Then, we propose an improved algorithm with several novel pruning techniques to improve the efficiency. Subsequently, we develop an algorithm which can compute the pareto-optimal bursting communities w.r.t.\ the parameters $l$ and $\delta$. Finally, we conduct comprehensive experiments using 9 real-life temporal networks, and the results demonstrate the efficiency, scalability and effectiveness of our algorithms.
\bibliographystyle{IEEEtran}
\bibliography{mybib}

\end{document}